\documentclass[onecolumn,floatfix,aps,nofootinbib]{revtex4}
\pdfoutput=1
\usepackage[dvips]{epsfig}
\usepackage[english]{babel}
\usepackage[utf8]{inputenc}
\usepackage{bbm}
\usepackage{verbatim}
\usepackage{array}
\usepackage{bm} 
\usepackage{amsmath}
\usepackage{yfonts}
\usepackage{amsthm}
\usepackage{amsmath,amscd}
\usepackage{pst-plot}
\usepackage{slashed} 
\usepackage{tikz-cd}
\usepackage{amsfonts}
\usepackage{graphicx}
\usepackage{amssymb}
\newtheorem{definition}{Definition} 
\usepackage{hyperref}
\usepackage{cases}
\usepackage{indentfirst} 
\usepackage[titletoc]{appendix}
\usepackage{subeqnarray}
\usepackage{setspace}
\usepackage{indentfirst} 

\usepackage{gensymb}
\usepackage{xcolor}
\usepackage{calrsfs}
\newtheorem{theorem}{Theorem}

\usepackage{wasysym}
\usepackage[all,cmtip]{xy}

    \usepackage{amsmath,amsfonts,amssymb,amsthm,mathrsfs,bbm,braket}

    \newtheorem{prop}{Proposition}[section]
    
    \newtheorem{lemma}{Lemma}[section]

    \numberwithin{equation}{section}

\usepackage{bbold}

\begin{document}

\title{Foundational aspects of spinor structures and exotic spinors} 

\author{J. M. Hoff da Silva} 
\email{julio.hoff@unesp.br}
\affiliation{Departamento de F\'isica, Universidade
Estadual Paulista, UNESP, Av. Dr. Ariberto Pereira da Cunha, 333, Guaratinguet\'a, SP,
Brazil.}

\begin{abstract}
Spinors are mathematical objects susceptible to the spacetime characteristics upon which they are defined. Not all spacetimes admit spinor structure; when it does, it may have more than one spinor structure, depending on topological properties. When more than one nonequivalent spinor structure is allowed in a given spacetime, the spinors resulting from the extra structures are called exotic. In this review, we revisit the topological conditions driving the discussion about the spacetime characteristics leading to the existence and (non)uniqueness of spinor structures in a relatively comprehensive manner, accounting for step-to-step demonstrations. In the sequel, we delve into the topologically corrected Dirac operator, explicitly obtaining it and exploring the physical consequences encoded in the exotic spinor dynamics. Finally, we overview early and recent achievements in the area, pointing out possible directions within this research program.    

\vspace{.35cm}
\begin{flushright}{\it In memory of Victor Hugo Assis Hoff Brait}\end{flushright}

\normalsize{\tableofcontents}
\end{abstract}		

\maketitle

\newpage

\section{Introduction}

Since the early work of Dirac \cite{dirac}, spinors (as the core construction of fermions) are frequently used to describe matter fields. The current understanding links spinors as objects carrying irreducible representations of the Spin group and entering as expansion coefficients of the respective quantum field \cite{wwei}. From the classical point of view, spinors are quite rich mathematical structures whose very existence is amalgamated with the underlying manifold structure. Spinors are sections of spinor bundles, structures that may or may not exist over a given manifold. When they exists, it can be non-unique, giving rise to nonequivalent structures that, in turn, have nonequivalent sections, the so-called exotic spinors. The construction of spinor structures via special lifts regards obstruction theory, more specifically, the Stiefel-Whitney classes. Such behavior of spinors attaches them to the manifold topology so that the spin connection is changed depending on the chosen section dynamics to be described.  

The above-mentioned aspect of spinor fields is more profound than it may sound at first sight. In fact, it is possible to implement nontrivial aspects of the base manifold in tensor field dynamics using appropriate dynamical equations of motion boundary conditions. Nevertheless, the situation of spinorial fields is much more involved. It reveals a genuinely umbilical relation with the manifold topology, disclosing prominent physical possibilities whose study sheds light on the topology (and vice versa).                

When regarded as an additional ingredient in the analysis of a physical concept, topology sometimes points to the necessity of its triviality, but from time to time, what calls attention to a given physical investigation are precisely the topological nontrivial aspects. Without any pretension of exhausting the theme in physics, it is pertinent to recall that the shape of the universe on a large scale is a long-standing issue \cite{haa}. Since the Hubble horizon is bound to the cosmological data, different topological setups cannot be ruled out by the present data. Moreover, a statistical treatment indicates that the cosmic microwave background temperature gradient can be better explained under the auspices of a multiply connected spacetime setup \cite{muni}. On a different level, nontrivial topology is also the target of study and analysis in the realm of condensed matter physics (see, for instance, \cite{mui}). 

Given the spinor theory richness, especially when considering base manifold peculiarities as nontrivial topology, it is desirable to give a complete account of the issue from the mathematical perspective, after which the physical consequences can be derived and appreciated. This is the purpose of this review, namely to give a step-to-step exposition of the spinorial mathematical theory underlining the definition and possibilities for different spinor fields, with an emphasis on their topological features, realize a complete derivation of exotic spinors as a consequence of the previous development, and evince some relevant points from the physical point of view. The idea is to focus mathematical preliminaries in the foundational part, presenting it in a structured manner, where the details of proofs (along with the chosen subtopics) rest on the exposition's idiosyncratic aspects. As well known, spinor theory may be presented in quite a complete mathematical fashion with an explicit relation with Clifford algebra via K\"ahler-Atiyah bundle (see Ref. \cite{dgf} for an excellent and comprehensive presentation of this perspective). We shall take a topological inclined exposition, which still requires some familiarity with fibre bundle theory, at the level of Ref. \cite{nak}. The impact of diffeomorphisms on spinors and spinorial structures is also somewhat analyzed. The physical consequences are then delineated as a consequence of the derived complete Dirac operator with a topological correction. In this part of the review, some new results are presented, such as the possibility of a quasi-particle behavior (in this context), a dispersion relation deviation, and additional topological correction in the effective vertex function.     

This paper is planned as follows: aimed at presenting a self-contained mathematical fundamental part (already starting from the fibre bundle concept), in Section II, we explore in detail the conditions under which spinor structures can be defined. The main line of exposition is to evince a necessary and sufficient condition for a principal bundle to support the spinor structure and then subvert it, in some precise sense, arriving at a clear topological obstruction. Going further, we focus on the uniqueness problem in a similar way, i.e., developing explicit criteria for uniqueness and further abandoning it from the topological point of view. This way constructed, the presentation aims to clarify the deep relation between spinor structures and the base manifold accommodating them. Section III brings a detailed account of the behavior of spinors under diffeomorphisms. We start reviewing in what precise sense it can be called trivial. In the sequel, we investigate the obstruction to this triviality in the same spirit as the previous section of the presentation. After an extensive analysis, the conclusion appears that in multiply-connected base manifolds, diffeomorphisms that are not connected to the identity permutes among spins structures. All the necessary demonstrations are presented in detail in this mathematically inclined part of the paper. Section IV is devoted to properly obtaining exotic spinors initially. We start by studying local sections defined as open sets belonging to the manifold covering. After some manipulation, a relation between usual spinors and their exotic counterpart is obtained at the same open set, allowing for a more general investigation. In dealing with the spinor dynamics, we get the complete Dirac operator explicit form without skipping the subtleties in deriving the topological (part of the) spin connection. Moving forward, we make explicit several relevant consequences of the topological correction to the exotic dynamics. Among them, we explore the exotic spinor dispersion relation and the spinorial current Gordon decomposition as physical consequences.      

We have postponed an overview of the field to the last section, correlating it here and there with the main text presentation and pointing out some directions for future research.

\section{Foundations - Part I: General mathematical structure}

This mathematically inclined section shall be split into two parts of study: the existence and (non)uniqueness of spinor structures. These aspects underline the consequences of further physical analysis.    

\subsection{Existence of spinor structures}

Consider $P_G$ a principal bundle with structure group given by (the topological group) $G$ and similarly for $P_\Gamma$. The spinor bundle and its relation with the base manifold topology $\mathcal{M}$ is established through the consideration of a continuous homomorphism $\rho:\Gamma\rightarrow G$ with kernel given by a discrete group $D\subset C(\Gamma)$, where $C(\Gamma)$ stands for the center of $\Gamma$. Additionally, consider a fibre map $\eta:P_\Gamma\rightarrow P_G$ commuting with (right) action of the structure group, that is, for $z_\Gamma\in P_\Gamma$ and $\gamma \in \Gamma$, we have 
\begin{eqnarray}    
\eta&:&\left. P_\Gamma\rightarrow P_G\right.\nonumber\\&&
\left. z_\Gamma\cdot\gamma\mapsto \eta(z_\Gamma\cdot\gamma)=\eta(z_\Gamma)\cdot \rho(\gamma).\right.
\end{eqnarray} The bundle $P_\Gamma$ endowed with $\eta$ as defined above is called a $\Gamma-$structure in $P_G$. Details apart, the reader could well prompt recognize $G$ as $SO(3)$, $\Gamma$ as $Spin(1,3)$ and $D=\mathbb{Z}_2$, and report directly the analysis to one-half representations of Lorentz group. It is convenient, however, to keep the study in general terms. In any case, we aim to explicitly and detailed review necessary and sufficient conditions under which a $\Gamma-$structure exists and how many nonequivalent $\Gamma-$structures are allowed.  

Let $\{U_\alpha\}$ be a simple open covering of $\mathcal{M}$ and take a local section system $\sigma^G_\alpha:U_\alpha\rightarrow P_G$. Intuitively, a path in the base manifold engenders a path in the bundle itself, so that one faces the necessity of the so-called transition functions $g_{\alpha\beta}:U_\alpha\cap U_\beta\rightarrow G$ defined, for $x\in U_\alpha\cap U_\beta$, by $\sigma^G_\beta(x)=\sigma^G_\alpha(x)\cdot g_{\alpha\beta}(x)$. Moreover, transition functions obey the consistency relation $g_{\alpha\beta}(x)g_{\beta\mu}(x)=g_{\alpha\mu}(x)$ for $x\in U_\alpha\cap U_\beta\cap U_\mu$, in terms of which one has
\begin{equation}
\sigma^G_\mu=\sigma^G_\beta g_{\beta\mu}\;\; \therefore \;\; \sigma^G_\mu=\sigma^G_\alpha g_{\alpha\beta}g_{\beta\mu},
\end{equation} for $x\in U_\alpha\cap U_\beta$, implying $\sigma^G_\mu=\sigma^G_\alpha g_{\alpha\mu}$, by employing the consistency relation, without which the bundle local pieces cannot be properly glued together. 

Let us elaborate on the necessity of the so-called lifting functions (hereafter denoted by $\gamma_{\alpha\beta}$). 

\begin{lemma}\label{adm}  
$P_G$ admits a $\Gamma-$structure if, and only if, there are continuous maps $\gamma_{\alpha\beta}: U_\alpha\cap U_\beta\rightarrow \Gamma$ such that i) $\gamma_{\alpha\beta}(x)\gamma_{\beta\mu}(x)=\gamma_{\alpha\mu}(x)$ for $x\in U_\alpha\cap U_\beta\cap U_\mu$ and ii) $\rho\cdot\gamma_{\alpha\beta}=g_{\alpha\beta}$. 
\end{lemma}
\begin{proof}
($\Rightarrow$) If $P_G$ admits a $\Gamma-$structure, then there is a map $\eta:P_\Gamma\rightarrow P_G$ (as previously mentioned). Take $\sigma_\alpha^\Gamma$ as a local section system for $P^\Gamma$ and implement the $\eta$ map by $\sigma_\alpha^G=\eta\cdot \sigma_\alpha^\Gamma$. Besides, consider $\gamma_{\alpha\beta}$ as the transition functions in $P_\Gamma$. Hence, for $x\in U_\alpha\cap U_\beta$, we have 
\begin{eqnarray} 
\sigma^G_\beta(x)=\eta(\sigma^{\Gamma}_\beta(x))\underbrace{=}_{i)}\eta(\sigma^\Gamma_\alpha(x)\gamma_{\alpha\beta}(x))=\eta(\sigma^\Gamma_\alpha(x))\rho(\gamma_{\alpha\beta}(x))\underbrace{=}_{ii)}\sigma^G_\alpha(x)g_{\alpha\beta}(x). 
\end{eqnarray}
($\Leftarrow$) Assuming the validity of $i)$ and $ii)$, we have $\sigma^G_\beta=\sigma^G_{\alpha}g_{\alpha\beta}=\sigma^G_\alpha\rho(\gamma_{\alpha\beta})$. On the other hand $\sigma^\Gamma_\beta=\sigma^\Gamma_\alpha\gamma_{\alpha\beta}$. Now set $\kappa:P^\Gamma\rightarrow P^G$ in such a way that $\kappa(\sigma^\Gamma_\beta)=\sigma^G_\beta$. Therefore $\kappa(\sigma^\Gamma_\beta)=\kappa(\sigma^\Gamma_\alpha\gamma_{\alpha\beta})=\sigma^G_\alpha\rho(\gamma_{\alpha\beta})$ and  then $\kappa(\sigma^\Gamma_\alpha\gamma_{\alpha\beta})=\kappa(\sigma^\Gamma_\alpha)\rho(\gamma_{\alpha\beta})$, rendering $\kappa$ a $\eta$-like map.  
\end{proof}

In a broad brush, given a map $f:X\rightarrow Y$ and another one $g:Z\rightarrow Y$, a lifting from $X$ to $Z$ is a map $h:X\rightarrow Z$ such that $f=g\circ h$. In the present case we say that $g_{\alpha\beta}$ are lifted to $\gamma_{\alpha\beta}$ by means of the following diagram:
\begin{equation}
\xymatrix{
	  &\Gamma \ar[rd]^\rho \\  
	U_\alpha\cap U_\beta \ar[ru]^{\gamma_{\alpha\beta}} \ar[rr]_{g_{\alpha\beta}} & & G \nonumber 
}
\end{equation} 

Now consider the following construction. Take $x\in U_\alpha\cap U_\beta\cap U_\mu\neq\emptyset$ and set $p_{\alpha\beta\mu}(x)=\gamma_{\beta\mu}(x)\gamma_{\alpha\mu}^{-1}(x)\gamma_{\alpha\beta}(x)$. Obviously, if a $\Gamma-$structure is valid, this element is the identity in $\Gamma$. Our point here is to envisage topological obstructions to the existence of a $\Gamma-$structure. It is straightforward to see that $\rho (p_{\alpha\beta\mu}(x))=\rho(\gamma_{\beta\mu}(x)\gamma_{\alpha\mu}^{-1}(x)\gamma_{\alpha\beta}(x))=g_{\beta\mu}(x)g_{\alpha\mu}^{-1}(x)g_{\alpha\beta}(x)$ and, recalling that $(\sigma^G_\alpha)^{-1}\sigma^G_\beta=g_{\alpha\beta}$ (and therefore $g_{\beta\alpha}=g_{\alpha\beta}^{-1}$), $\rho(p_{\alpha\beta\mu}(x))=g_{\beta\mu}(x)g_{\mu\alpha}(x)g_{\alpha\beta}(x)=e$ via consistency conditions. Thus $p_{\alpha\beta\mu}(x)\in \ker(\rho)=D$. Since $D$ is a discrete group, $p_{\alpha\beta\mu}$ is constant, defining a 2-cochain $p$:
\begin{eqnarray}
p&:&\left. U_\alpha\cap U_\beta\cap U_\mu\rightarrow D\right. \nonumber\\&&
\left. x\mapsto p(\alpha,\beta,\mu)=p_{\alpha\beta\mu}(x).\right.
\end{eqnarray}

\begin{lemma}
The 2-cochain $p$ is a cocycle. 
\end{lemma}
\begin{proof}
Let $\partial$ be the coboundary operator. Then (omitting the base manifold point coordinate)
\begin{eqnarray}
(\partial p)(\alpha,\beta,\mu,\nu)=p(\beta,\mu,\nu)p^{-1}(\alpha,\mu,\nu)p(\alpha,\beta,\nu)p^{-1}(\alpha,\beta,\mu)=p_{\beta\mu\nu}p^{-1}_{\alpha\mu\nu}p_{\alpha\beta\nu}p^{-1}_{\alpha\beta\mu}.
\end{eqnarray} Recall that $D$ belongs to $C(\Gamma)$ being, thus, abelian. Hence $p(\beta,\mu,\nu)p^{-1}(\alpha,\mu,\nu)=p^{-1}(\alpha,\mu,\nu)p(\beta,\mu,\nu)$ and taking into account that $p_{\alpha\beta\mu}=\gamma_{\beta\mu}\gamma_{\alpha\mu}^{-1}\gamma_{\alpha\beta}$, we are left with 
\begin{eqnarray}
(\partial p)(\alpha,\beta,\mu,\nu)=(\gamma_{\mu\nu}\gamma^{-1}_{\alpha\nu}\gamma_{\alpha\mu})^{-1}(\gamma_{\mu\nu}\gamma^{-1}_{\beta\nu}\gamma_{\beta\mu})(\gamma_{\beta\nu}\gamma^{-1}_{\alpha\nu}\gamma_{\alpha\beta})(\gamma_{\beta\mu}\gamma^{-1}_{\alpha\mu}\gamma_{\alpha\beta})^{-1},
\end{eqnarray} or, equivalently,
\begin{equation}
(\partial p)(\alpha,\beta,\mu,\nu)=\gamma^{-1}_{\alpha\mu}\gamma_{\alpha\nu}\underbrace{\gamma^{-1}_{\mu\nu}\gamma_{\mu\nu}}_{e_{\Gamma}}\gamma^{-1}_{\beta\nu}\gamma_{\beta\mu}(\gamma_{\beta\nu}\gamma^{-1}_{\alpha\nu}\gamma_{\alpha\beta})(\gamma_{\beta\mu}\gamma^{-1}_{\alpha\mu}\gamma_{\alpha\beta})^{-1}.
\end{equation} Taking again into account that $p\in D\subset C(\Gamma)$, the $\gamma_{\beta\mu}$ commutes with the last two factors in parenthesis. So, the above expression reads
\begin{equation}
(\partial p)(\alpha,\beta,\mu,\nu)=\gamma^{-1}_{\alpha\mu}\gamma_{\alpha\nu}\gamma^{-1}_{\beta\nu}(\gamma_{\beta\nu}\gamma^{-1}_{\alpha\nu}\gamma_{\alpha\beta})(\gamma_{\beta\mu}\gamma^{-1}_{\alpha\mu}\gamma_{\alpha\beta})^{-1}\gamma_{\beta\mu}
\end{equation} which is $(\partial p)(\alpha,\beta,\mu,\nu)=e_{\Gamma}$. 
\end{proof}
Just as a remark, notice that for most practical situations (in which the $\Gamma-$structure is the $SO(1,3)$ lifting to $Spin(1,3)$) $p\in \mathbb{Z}_2$ and therefore a negative $p$ is also in order. This can be accomplished by chosen $\tilde{\gamma}_{\alpha\beta}=-\gamma_{\alpha\beta}$. Hence, calling $f(\alpha,\beta)=\pm 1$ the sign of $\tilde{\gamma}_{\alpha\beta}$ we can write $\tilde{p}(\alpha,\beta,\mu)=\tilde{\gamma}_{\beta\mu}\tilde{\gamma}^{-1}_{\alpha\mu}\tilde{\gamma}_{\alpha\beta}=p(\alpha,\beta,\mu)f(\beta,\mu)f^{-1}(\alpha,\mu)f(\alpha,\beta)$, which amounts out to be $\tilde{p}=p\partial f(\alpha,\beta)$, that is $p$ changes by a coboundary. Therefore, $p$ and $\tilde{p}$ are elements of the same class. This construction proves the independence of $[p]$ concerning the lift.   

The above reasoning allows the recognition of $[p]$ as an element of $\check{H}^2(\mathcal{M},D)$. The obstruction character of $[p]$ is somewhat evident: if, and only if, $[p]$ is trivial, the lifting functions satisfy the consistency condition, allowing a well-behaved relation between the base manifold covered by open sets and the corresponding local sections. In any case, this result is formally stated through the following theorem.
\begin{theorem}
A given principal bundle admits a $\Gamma-$structure if, and only if, $[p]$ is trivial. 
\end{theorem}
\begin{proof}
($\Rightarrow$) Assuming the existence of a $\Gamma-$structure, by Lemma \ref{adm}, there are lifting functions satisfying $\gamma_{\alpha\beta}(x)\gamma_{\beta\mu}(x)=\gamma_{\alpha\mu}(x)$. Therefore, applying $\gamma^{-1}_{\alpha\mu}\gamma_{\alpha\beta}$ from the right one has $\gamma_{\alpha\beta}p_{\alpha\beta\mu}=\gamma_{\alpha\beta}$. Hence, trivially, $[p]=e$. 

($\Leftarrow$) Assuming $[p]=e$, it is possible to assert the existence of a 1-cochain, say $q_{\alpha\beta}$, such that $p(\alpha,\beta,\mu)=q_{\beta\mu}q^{-1}_{\alpha\mu}q_{\alpha\beta}$, for $\partial p(\alpha,\beta,\mu,\nu)=(q_{\mu\nu}q^{-1}_{\beta\nu}q_{\beta\mu})(q_{\mu\nu}q^{-1}_{\alpha\nu}q_{\alpha\mu})^{-1}(q_{\beta\nu}q^{-1}_{\alpha\nu}q_{\alpha\beta})(q_{\beta\mu}q^{-1}_{\alpha\mu}q_{\alpha\beta})^{-1}=e$, since $q_{\alpha\beta}$ takes values on the abelian group $D$. 

Now set maps
\begin{eqnarray}
 m_{\alpha\beta}&:&\left. U_\alpha\cap U_\beta\rightarrow \Gamma\right.\nonumber\\&&
 \left. x\mapsto m_{\alpha\beta}=\gamma_{\alpha\beta}(x)q_{\alpha\beta}^{-1},\right.  
\end{eqnarray} where $\gamma_{\alpha\beta}$ are defined as before. Remembering the definition of $p_{\alpha\beta\mu}$, it can be readily verified that 
\begin{equation}
m_{\alpha\beta}m_{\beta\mu}=\gamma_{\alpha\beta}q^{-1}_{\alpha\beta}p_{\alpha\beta\mu}\gamma_{\alpha\beta}^{-1}\gamma_{\alpha\mu} q_{\beta\mu}^{-1}
\end{equation} and, since $q_{\alpha\beta}\in D$, $p_{\alpha\beta\gamma}\in D$, and $D$ belongs to the center of $\Gamma$, terms can be interchanged such that  
\begin{equation}
m_{\alpha\beta}m_{\beta\mu}=q^{-1}_{\alpha\beta}p_{\alpha\beta\mu}\gamma_{\alpha\mu}q^{-1}_{\beta\mu}. 
\end{equation} By inserting the identity $q_{\alpha\mu}^{-1}q_{\alpha\mu}$ between $\gamma$ and $q$ in the above expression and recognizing $m_{\alpha\mu}=\gamma_{\alpha\mu} q_{\alpha\mu}^{-1}$ we are left with 
\begin{equation}\label{la}
m_{\alpha\beta}m_{\beta\mu}=m_{\alpha\mu}q^{-1}_{\alpha\beta}p_{\alpha\beta\mu}q_{\alpha\mu}q_{\beta\mu}^{-1}.
\end{equation} Finally, from $p(\alpha,\beta,\mu)=q_{\beta\mu}q^{-1}_{\alpha\mu}q_{\alpha\beta}$, the expression (\ref{la}) reduces to the $m_{\alpha\beta}m_{\beta\mu}=m_{\alpha\mu}$. Besides, $\rho\circ m_{\alpha\beta}=\rho(\gamma_{\alpha\beta}q^{-1}_{\alpha\beta})=\rho(\gamma_{\alpha\beta}q_{\beta\alpha})$ and, since $\rho$ is a homomorphism $\rho\cdot m_{\alpha\beta}=\rho(\gamma_{\alpha\beta})\rho(q_{\beta\alpha})=g_{\alpha\beta}$ because $q_{\beta\alpha}\in D=\ker(\rho)$.   
\end{proof}

From a mathematical point of view, it is important to study the relationship between fibre bundles and their impact on the obstruction class. The fact is that a given map $\phi:\bar{P}_G\rightarrow P_G$ engendering fibre homeomorphisms leads to a relation between the obstruction class in both bundles usually called {\it naturality}. For our purposes, we shall study this aspect not only by completeness but also because it introduces relevant procedures to study the behavior of spinors under diffeomorphisms. 

Consider two fibre bundles $(\bar{P}_G,\bar{\pi},\bar{\mathcal{M}},G)$ and $(P_G,\pi,\mathcal{M},G)$, and call $\beta:\bar{\mathcal{M}}\rightarrow \mathcal{M}$ establishing the base spaces homeomorphism. The next lemma establishes the result. 
\begin{lemma}
	Denote by $\beta^{\ddag}:\check{H}(\mathcal{M},D)\rightarrow \check{H}(\bar{\mathcal{M}},D)$ the homomorphism induced by $\beta$. Then $[p](\bar{P}_G,\rho)=\beta^\ddag\cdot [p](P_G,\rho)$.
\end{lemma}
\begin{proof}
	Generally, we deal exclusively with spaces where every open covering admits a simple refinement. Consider, then, two simple coverings $\{\bar{U}_\varpi\}_{\varpi\in \aleph}$ of $\bar{\mathcal{M}}$ and $\{\bar{U}_\alpha\}_{\alpha\in I}$ of $\mathcal{M}$ such that for some $\alpha$ we have $\beta (\bar{U}_\varpi)\subset U_\alpha$. Then it is possible to consider a function $\tau:\aleph\rightarrow I$ such that $\beta(\bar{U}_\varpi)\subset U_{\tau(\varpi)}$.
		
	From a system of local sections $\sigma_\alpha:U_\alpha\rightarrow P$, establish a system of local sections for $\bar{P}$ by
	\begin{equation}\label{sac}
	\bar{\sigma}_\alpha(\bar{x})=\phi^{-1}|_{\bar{x}}\cdot \sigma_{\tau(\alpha)}(\beta|_{\bar{x}}),
	\end{equation} where $\phi|_{\bar{x}}$ denotes the restriction of $\phi$ to the fibre on $\bar{x}\in \bar{U}_{\alpha}$, with a analog meaning to $\beta|_{\bar{x}}$. Notice that the transition functions may be related by (\ref{sac}). In fact, $\bar{g}_{\alpha\beta}(\bar{x})=\bar{\sigma}^{-1}_\alpha\cdot\bar{\sigma}_\beta$ leads to (using (\ref{sac}))
	\begin{equation}
	\bar{g}_{\alpha\beta}(\bar{x})=(\phi^{-1}|_{\bar{x}}\cdot \sigma_{\tau(\alpha)}(\beta|_{\bar{x}}))^{-1}\cdot(\phi^{-1}|_{\bar{x}}\cdot \sigma_{\tau(\beta)}(\beta|_{\bar{x}})),
	\end{equation} or
	\begin{equation}
	\bar{g}_{\alpha\beta}(\bar{x})=\sigma^ {-1}_{\tau(\alpha)}(\beta|_{\bar{x}})\cdot\underbrace{\phi|_{\bar{x}}\cdot \phi^{-1}|_{\bar{x}}}_{Id_{P}}\cdot \sigma_{\tau(\beta)}(\beta|_{\bar{x}}).
	\end{equation} Therefore $\bar{g}_{\alpha\beta}(\bar{x})=\sigma^{-1}_{\tau(\alpha)}\cdot\sigma_{\tau(\beta)}(\beta|_{\bar{x}})=g_{\tau(\alpha)\tau(\beta)}(\beta|_{\bar{x}})$, with $\bar{x}\in \bar{U}_\alpha\cap\bar{U}_\beta$. Now, recall that $\gamma_{\alpha\beta}:U_\alpha\cap U_\beta\rightarrow \Gamma$ are the lifts of $g_{\alpha\beta}\in G$ and define $\bar{\gamma}_{\alpha\beta}(\bar{x})=\gamma_{\tau(\alpha)\tau(\beta)}(\beta|_{\bar{x}})$. Hence
	\begin{equation}
	\rho\cdot\bar{\gamma}_{\alpha\beta}(\bar{x})=\rho\cdot\gamma_{\tau(\alpha)\tau(\beta)}(\beta|_{\bar{x}})=g_{\tau(\alpha)\tau(\beta)}(\beta|_{\bar{x}})=\bar{g}_{\alpha\beta}(\bar{x})
	\end{equation} and, in an analog fashion to the $\gamma_{\alpha\beta}$ case, $\bar{\gamma}_{\alpha\beta}$ are the lifts of $\bar{g}_{\alpha\beta}$. This last result opens the possibility for considering the cocycle obstruction for $\bar{P}_G$ so that 
	\begin{equation}
	\bar{p}(\alpha,\beta,\mu)=\bar{\gamma}_{\beta\mu}(\bar{x})\bar{\gamma}^{-1}_{\alpha\mu}(\bar{x})\bar{\gamma}_{\alpha\beta}(\bar{x})=\gamma_{\tau(\beta)\tau(\mu)}(\beta|_{\bar{x}})\gamma^{-1}_{\tau(\alpha)\tau(\mu)}(\beta|_{\bar{x}})\gamma_{\tau(\alpha)\tau(\beta)}(\beta|_{\bar{x}}),
	\end{equation} which means $\bar{p}(\alpha,\beta,\mu)=p(\tau(\alpha),\tau(\beta),\tau(\mu))(\beta|_{\bar{x}})$ and, therefore, $[p]_{(\bar{P}_G,\rho)}=\beta^\ddag\cdot[p]_{(P_G,\rho)}$. 
\end{proof}

It is straightforward to generalize the last result to $n-1$ ($n\in \mathbb{N}^*$) sequential homeomorphisms between $(\bar{P_i},\bar{\pi_i},\bar{\mathcal{M}_i},G)$ for $i=1,\cdots ,n$. 
 \begin{equation}
\xymatrixcolsep{4pc}\xymatrix{
		\cdots \bar{P_{2G}} \ar[d]_-{\bar{\pi}_2} 
	\ar[r]^-{\phi_2} &\bar{P_{1G}} \ar[d]_-{\bar{\pi}_1} \ar[r]^-{\phi_1} & \bar{P_{G}} \ar[d]_-{\pi} \\
	\cdots \bar{\mathcal{M}}_2\ar[r]^-{\beta_2}& \bar{\mathcal{M}}_1 \ar[r]^-{\beta_1}& \bar{\mathcal{M}}
}
\end{equation} 

In fact, setting $\phi_i:\bar{P_i}_G\rightarrow \bar{P_{i-1}}_G$ ($\bar{P_0}_G\equiv\bar{P}_G$), $\beta_i:\bar{\mathcal{M}_i}\rightarrow \bar{\mathcal{M}_{i-1}}$ ($\bar{\mathcal{M}_0}\equiv\bar{\mathcal{M}}$), $\beta_i^\ddag:\check{H}^2(\bar{\mathcal{M}_{i-1}},D)\rightarrow \check{H}^2(\bar{\mathcal{M}_{i}},D)$ for $i=1,\cdots ,n$ (see the above diagram for sequential homeomorphisms), the relation between transition functions in the $n-$th bundle and $(\bar{P},\bar{\pi},\bar{\mathcal{M}},G)$ is given by
\begin{equation}
\bar{g_n}_{\alpha\beta}(\bar{x}_n)=g_{\tau_1\cdots \tau_{n-1}(\tau_n(\alpha))\tau_1\cdots \tau_{n-1}(\tau_n(\beta))}(\beta_1\circ\cdots\circ\beta_n|_{\bar{x}_n}),
\end{equation} where $\tau_i:\aleph_i\rightarrow\aleph_{i-1}$ ($\aleph_0\equiv I$) and $\bar{x}_i\in \bar{U_i}_\alpha\cap\bar{U_i}_\beta$ $\forall i$. Lifting functions can be concatenated such that $\rho\cdot\bar{\gamma_n}_{\alpha\beta}=\bar{g_n}_{\alpha\beta}(\bar{x}_n)$, resulting in the following obstruction classes relation 
\begin{equation}
[p]_{(\bar{P_n},\rho)}=\beta_n^\ddag\cdot\,\cdots\,\beta_1^\ddag [p]_{(P,\rho)}.
\end{equation} Anyhow, it is possible to observe the pertinence of naturality: being two fibre bundles congruent, their topology is equivalent (in the homeomorphic mapping sense). Therefore, their obstruction classes are also congruent.  

From a physical point of view, it is relevant to say some words about the obstruction class that was just analyzed and its connection with spacetime. Geroch constructed this bridge in Refs. \cite{G1} and \cite{G2}. In a framework involving time and space orientability, the work presented in Ref. \cite{G1} shows that a given spacetime admits spinor structure if, and only if, it admits a global system of orthonormal tetrads. In practical terms, the link between the very existence of spinor structures and orthonormal tetrads is necessary since describing spinor dynamics in generic spacetimes requires globally defined tetrads to compose the spin connection. Nevertheless, the results of \cite{G1} make the status of this relation much more deep. In \cite{G2}, the existence of spinor fields is connected to a spacetime coming from a well-posed Cauchy problem.  

An explicit connection between topological and physical considerations for the spinors was performed in Ref. \cite{lee}. The general argumentation line can be drawn as follows: we have said that $[p]\in \check{H}(\mathcal{M},D\simeq \mathbb{Z}_2)$, is nontrivial, is an obstruction to the admittance of $\Gamma-$structures by principal bundles. This obstruction class is nothing but the second Stiefel-Whitney class, $w_2$, of $\mathcal{M}$ \cite{nak,bono}. The $n$-th Stiefel-Whitney class ($w_n$) of a base manifold obstructs the existence of a vector bundle over it with coefficients in $\mathbb{Z}_2$. In four dimensions, there is a very important result \cite{st} connecting the triviality of $w_n=0$ ($n=1,2,3,4$) and the existence of a continuous field of orthogonal $(4-(n-1))$-frames over the manifold $n$-dimensional skeleton $sk_n(\mathcal{M})$. The manifold is assumed to be orientable (as it shall play the role of a genuine spacetime) and therefore $w_1=0$. Hence, continuous tetrads can be placed at $sk_1(\mathcal{M})$, or each line (one-surface) point of $\mathcal{M}$, i.e., without discrete reflections. As discussed, $w_2=0$ and orthogonal triads can be continuously placed at the plane (two-surface) points of $\mathcal{M}$ by admitting spinor structures. When $w_1$ and $w_2$ are both trivial, then $w_3$ must also be trivial\footnote{Using the so-called Steenrod square, $Sq^r:\check{H}^m(\mathcal{M},\mathbb{Z}_2)\rightarrow \check{H}^{r+m}(\mathcal{M},\mathbb{Z}_2)$ defined in a cohomology ring, it is possible to envisage this result, following the theorem \cite{mili,milum} (whose simplified statement here will be done without proof):
	\begin{theorem}
		The cohomology class $Sq^r(w_p)$ can be expressed as 
		\begin{eqnarray}
		Sq^r(w_p)=w_rw_p+\binom{p-r}{1}w_{r-1}w_{p+1}+\binom{p-r+1}{2}w_{r-2}w_{p+2}+\cdots+\binom{p-1}{p}w_0w_{p+r}. 
		\end{eqnarray}
	\end{theorem} Thus, we are left with $Sq^1(w_2)=w_1w_2+w_0w_3$. Nevertheless, $w_0=1$ by definition, $\mathcal{M}$ is orientable ($w_1$ is vanishing) and admits spinor structures ($w_2$ is also null), hence we arrive at $w_3=0$. We assume $\mathcal{M}$ paracompact for this simplified theorem version. Therefore, a coincidence exists between singular and \v{C}ech cohomologies, and the cup product redounds in a usual product.} \cite{bre}. Moreover (see \cite{st}), for noncompact $\mathcal{M}$ all \v{C}ech cohomology group $\check{H}^4(\mathcal{M},D)$ is trivial and therefore so it is $w_4\in \check{H}^4(\mathcal{M},D)$. For compact base manifolds, the triviality of $w_4$ is obtained by vanishing the Euler characteristic class. The vanishing of $w_3$ and $w_4$ are related to chirality \cite{o1} and causality \cite {o2}, respectively.

It is instructive to call attention to Refs. \cite{r1,r2} where modified Stiefel-Whitney classes are proposed to give an account of spin structures in arbitrary dimensions (signatures, and possibly reflections) obstructions.  

\subsection{(Non)uniqueness of spinor structures}

So far, we have been concerned about the existence of spinor structures and, therefore, spinors in a given space. Now, we shall focus on the uniqueness (or lack thereof) aspects. Consider again, for such, a principal fibre bundle $P_G$ with base manifold $\mathcal{M}$ and an homomorphism $\lambda:G\rightarrow G_\lambda$.   

\begin{prop}
	The homomorphism $\lambda:G\rightarrow G_\lambda$ determines a principal fibre bundle $P_{G_\lambda}$ with $\mathcal{M}$ as the base manifold.   
\end{prop}
\begin{proof}
	Start choosing a covering $\{U_\alpha\}$ of $\mathcal{M}$ and define $\sigma_\alpha$ and $g_{\alpha\beta}\in G$ as the local sections and transitions functions, respectively. The action of $\lambda$ is naturally given by $g^{(\lambda)}_{\alpha\beta}=\lambda g_{\alpha\beta}$. Therefore, for $x\in U_\alpha\cap U_\beta \cap U_\kappa$ and bearing in mind that $\lambda$ is a homomorphism\footnote{In general, for an group isomorphism $\lambda$, being $*$ the product in the domain group and $\cdot$ the product in the arriving one, we have $\lambda(u*v)=\lambda(u)\cdot\lambda(v)$. Hence, making $v=u^{-1}$ and recalling that $\lambda(u^{-1})=\lambda^{-1}(u)$, one directly arrives at $e_*\in \ker(\lambda)$.}, we have
	\begin{equation}
	g^{(\lambda)}_{\alpha\beta}(x)g^{(\lambda)}_{\beta\kappa}(x)g^{(\lambda)}_{\kappa\alpha}(x)=\lambda g_{\alpha\beta}\lambda g_{\beta\kappa}\lambda g_{\kappa\alpha}=\lambda (g_{\alpha\beta}g_{\beta\kappa}g_{\kappa\alpha})=\lambda(e_G)=e_{G_{\lambda}}.
	\end{equation} Therefore, there is a principal fibre bundle $P_{G_\lambda}$ with a local section system such that $g^{(\lambda)}_{\alpha\beta}$ are the corresponding transition functions.  
 \end{proof} 

This previous result establishes an extension that can be generalized employing central homomorphisms. Take, as usual, $\rho:\Gamma\rightarrow G$, with $\ker(\rho)=D$ and define $\tilde{\rho}:\Gamma_\lambda\rightarrow G_\lambda$ with $\ker(\tilde{\rho})=D_\lambda$. 
\begin{definition}
	The central homomorphisms $\rho$ and $\tilde{\rho}$ are said to be related if there exists a continuous map $\tilde{\lambda}$ such that
	
	$i)$ The diagram below commutes   
 \begin{equation}
\xymatrixcolsep{4pc}\xymatrix{
	\Gamma \ar[d]_-{\rho} 
	\ar[r]^-{\tilde{\lambda}} &\Gamma_{\lambda} \ar[d]_-{\tilde{\rho}} \\
	G\ar[r]^-{\lambda}& G_\lambda
}
\end{equation}

	$ii)$ For every $d\in D$ and $g\in \Gamma$, $\tilde{\lambda}(d\cdot g)=\tilde{\lambda}(d)\tilde{\lambda}(g)$.
\end{definition} Observe, from the diagram, that $\tilde{\lambda}=\tilde{\rho}^{-1}\cdot\lambda\cdot\rho$ and, hence, for $d\in D$ we have $\tilde{\rho}\tilde{\lambda}(d)=\lambda\cdot \rho(d)=\lambda(e_G)=e_{G_\lambda}$, since $\lambda$ is a homomorphism. Therefore, $\tilde{\lambda}(d)\in D_\lambda$ restricts to a homomorphism from $D$ to $D_{\lambda}$. The above definition and comment are sufficient for establishing an important partial result. 

\begin{prop} 	   	
	Let $P_{G_\lambda}$ be an extension of $P_G$ and consider related central homomorphisms $\rho$ and $\tilde{\rho}$. If $\gamma_{\alpha\beta}$ are the lifts of $g_{\alpha\beta}$, then $\gamma_{\alpha\beta}^{\lambda}=\tilde{\lambda}\cdot \gamma_{\alpha\beta}$ lifts $g_{\alpha\beta}^\lambda$. Besides, for $x\in U_\alpha\cap U_\beta\cap U_\kappa$ 
	\begin{equation} \Theta_{\alpha\beta\kappa}(x)=\tilde{\lambda}(\gamma_{\alpha\beta}(x))\tilde{\lambda}(\gamma_{\beta\kappa}(x))\tilde{\lambda}(\gamma_{\alpha\beta}(x)\gamma_{\beta\kappa}(x))^{-1}\in D_\lambda.
	\end{equation}  
\end{prop}
\begin{proof}
	Note that $\tilde{\rho}(\gamma^\lambda_{\alpha\beta})=\tilde{\rho}(\tilde{\lambda}\gamma_{\alpha\beta})=\lambda\rho\gamma_{\alpha\beta}$ (see the previous definition diagram). Thus $\tilde{\rho}\gamma^\lambda_{\alpha\beta}=g^\lambda_{\alpha\beta}$ and $\gamma_{\alpha\beta}^{\lambda}$ lifts $g^\lambda_{\alpha\beta}$.
	
	In the sequel, observe that (omitting the point) $e_\Gamma=\gamma_{\alpha\beta}\gamma_{\beta\kappa}(\gamma_{\alpha\beta}\gamma_{\beta\kappa})^{-1}\in D$ and, therefore, $\tilde{\lambda}$ acts as an isomorphism over it. In particular, it is possible to write $\Theta_{\alpha\beta\kappa}=\tilde{\lambda}[\gamma_{\alpha\beta}\gamma_{\beta\kappa}(\gamma_{\alpha\beta}\gamma_{\beta\kappa})^{-1}]$, so that $\tilde{\rho}\Theta_{\alpha\beta\kappa}=\lambda\rho[\gamma_{\alpha\beta}\gamma_{\beta\kappa}(\gamma_{\alpha\beta}\gamma_{\beta\kappa})^{-1}]$. Since $\rho$ is an isomorphism, we are left with 
	\begin{eqnarray}
	\tilde{\rho}\Theta_{\alpha\beta\kappa}=\lambda[\rho(\gamma_{\alpha\beta})\rho(\gamma_{\beta\kappa})\rho(\gamma_{\alpha\beta}\gamma_{\beta\kappa})^{-1}]=\lambda[g_{\alpha\beta}g_{\beta\kappa}(g_{\alpha\beta}g_{\beta\kappa})^{-1}]=\lambda(e_G)=e_{G_{\lambda}}. 
	\end{eqnarray} Hence $\Theta_{\alpha\beta\kappa}\in D_\lambda$. 
\end{proof}	
     
In view of the above result, $\Theta_{\alpha\beta\kappa}$ define a 2-cochain $\Theta$ with values in $D_\lambda$. Since, at this stage, we are concerned with the fibre bundle extension via $\lambda$ and generalizations, it is relevant to study how (and if) the obstructions of a given bundle are related to the other. The next proposition cares about this point. 
\begin{prop}
         Let $p$ and $p^\lambda$ be 2-cocycles to $P_G$ and $P_{G_\lambda}$, respectively. Then $p^\lambda(\alpha,\beta,\kappa)=\Theta(\alpha,\beta,\kappa)\tilde{\lambda}p(\alpha,\beta,\kappa)$.
\end{prop}         
\begin{proof}
We start noticing that $p^{-1}_{\alpha\beta\gamma}=\gamma_{\alpha\beta}^{-1}\gamma_{\alpha\kappa}\gamma_{\beta\kappa}^{-1}$ and, hence, $\gamma_{\alpha\kappa}=p^{-1}_{\alpha\beta\kappa}\gamma_{\alpha\beta}\gamma_{\beta\kappa}$, since $p_{\alpha\beta\kappa}\in C(\Gamma)$. Acting with $\tilde{\lambda}$  from the left and recalling that $p_{\alpha\beta\kappa}\in D$ and $\gamma_{\alpha\beta}\gamma_{\beta\kappa}\in \Gamma$, we have 
\begin{equation}
\tilde{\lambda}(\gamma_{\alpha\kappa})=\tilde{\lambda}(p^{-1}_{\alpha\beta\kappa})\tilde{\lambda}(\gamma_{\alpha\beta}\gamma_{\beta\kappa}). \label{aci}
\end{equation} On the other hand, $\gamma^{\lambda}_{\alpha\kappa}=(p^\lambda)^{-1}_{\alpha\beta\kappa}\gamma^\lambda_{\alpha\beta}\gamma^\lambda_{\beta\kappa}=(p^\lambda)^{-1}_{\alpha\beta\kappa}\tilde{\lambda}(\gamma_{\alpha\beta})\tilde{\lambda}(\gamma_{\beta\kappa})$. With the aid of Eq. (\ref{aci}) this last equation reads
\begin{equation}
(p^\lambda)^{-1}_{\alpha\beta\kappa}\tilde{\lambda}(\gamma_{\alpha\beta})\tilde{\lambda}(\gamma_{\beta\kappa})=\tilde{\lambda}(p^{-1}_{\alpha\beta\kappa})\tilde{\lambda}(\gamma_{\alpha\beta}\gamma_{\beta\kappa}),\label{acb}
\end{equation} from which 
\begin{equation}
p^\lambda_{\alpha\beta\kappa}=\tilde{\lambda}(\gamma_{\alpha\beta})\tilde{\lambda}(\gamma_{\beta\kappa})\tilde{\lambda}(\gamma_{\alpha\beta}\gamma_{\beta\kappa})^{-1}\tilde{\lambda}^{-1}(p^{-1}_{\alpha\beta\kappa}).
\end{equation} Finally, recalling the $\Theta_{\alpha\beta\kappa}$ definition (and observing that the last term is nothing but $\tilde{\lambda}(p_{\alpha\beta\kappa})$), the proposition is proved. \end{proof}

Before moving on, let us evince a direct consequence of the proposition just proved. Calling $\psi^\ddag$ the induced homomorphism from $\check{H}^2(\mathcal{M},D)$ to $\check{H}^2(\mathcal{M},D_\lambda)$, it is readily obtained that the obstruction classes of $P_G$ and $P_{G_\lambda}$ are related by $[p^\lambda]_{(P_{G_\lambda},\tilde{\rho})}=\Theta\cdot\psi^\ddag[p]_{(P_G,\rho)}$. Notice that in the presentation, we have said that $\tilde{\lambda}$ is a homomorphism when restricted to the center of $\Gamma$ and $\Gamma_\lambda$. Had we defined $\tilde{\lambda}$ as a homomorphism for the entire $\Gamma$, then $\Theta$ would collapse to the identity and the obstruction classes relation would simplify to $[p^\lambda]_{(P_{G_\lambda},\tilde{\rho})}=\psi^\ddag[p]_{(P_G,\rho)}$.  

We can now study the conditions under which the $\Gamma-$structure is unique. Start admitting a principal fibre bundle $P_G$ with two $\Gamma-$structures $(P_1,\eta_1)$ and $(P_2,\eta_2)$. If $\sigma^G_\alpha$ performs local sections for $P_G$ (defined over a simple cover), then there exist local sections $\sigma^1_\alpha$ ($\sigma^2_\alpha$) for $P_1$ ($P_2$) such that $\eta_1\cdot \sigma^1_\alpha=\sigma_\alpha$ ($\eta_2\cdot \sigma^2_\alpha=\sigma_\alpha$). Call $\gamma^i_{\alpha\beta}:U_\alpha\cap U_\beta\rightarrow\Gamma^i$ ($i=1,2$) the transition functions and define special maps $\delta_{\alpha\beta}:U_\alpha\cap U_\beta\rightarrow \Gamma^1$ by 
\begin{equation}
\delta_{\alpha\beta}(x)=\gamma^1_{\alpha\beta}(x)(\gamma^2_{\alpha\beta})^{-1}(x), \hspace{.2cm} x\in U_\alpha\cap U_\beta.\label{spc}
\end{equation}It is important to say some words about the map above. Notice that it engenders a correlation between the $\Gamma-$structures in a rather special way: it passes through a point in the base manifold. This is not new in some sense: recall the discussion around Theorem 1. The encoded results depend on the realization (or not) of the consistency conditions by the lifts. Base manifold obstructions to this realization forbid the very existence of spinor structures. Here, we shall face a similar aspect: the base manifold shall dictate the ensuing results. Some base manifolds will be such that $\delta_{\alpha\beta}$ are trivial (equal to de identity) and hence, from (\ref{spc}), the pretense different transition functions $\gamma^1$ and $\gamma^2$ are related.

On the other hand, obstructions to the $\delta_{\alpha\beta}$ triviality lead to genuinely different spin structures. These aspects show the umbilical relation between spinors and the spacetime upon which they are defined. In the sequel, we shall formalize the last part of these comments with several propositions.   
\begin{prop}
  The map $\delta_{\alpha\beta}:U_\alpha\cap U_\beta\rightarrow \Gamma^1$ given in (\ref{spc}) is a cochain. 
\end{prop}  
\begin{proof}\label{vis}
	Let $\bar{\lambda}:\Gamma^2\rightarrow \Gamma^1$ ($\gamma^1_{\alpha\beta}=\bar{\lambda}\gamma^2_{\alpha\beta}$) a homomorphism restricted to the kernel. Applying the central homomorphism, we have $\rho[\delta_{\alpha\beta}]=\rho[\gamma^1_{\alpha\beta}(\gamma^2_{\alpha\beta})^{-1}]=\rho[\bar{\lambda}\underbrace{\gamma^2_{\alpha\beta}(\gamma^2_{\alpha\beta})^{-1}}_{e_{\Gamma^2}}]$, so that $\rho[\delta_{\alpha\beta}]=\rho[e_{\Gamma^1}]=e_{G}$, and therefore $\delta_{\alpha\beta}(x)\in D$. 
\end{proof}	
\begin{prop}
	The cochain $\delta_{\alpha\beta}$ is a cocycle. 
\end{prop}	
\begin{proof}
	Take $x\in U_\alpha\cap U_\beta\cap U_\kappa$ and compute the cobordism of $\delta$ by $\partial\delta(\alpha,\beta,\kappa)=\delta_{\beta\kappa}\delta^{-1}_{\alpha\kappa}\delta_{\alpha\beta}$; it yields
	\begin{equation}
	\partial\delta(\alpha,\beta,\kappa)=[\gamma^1_{\beta\kappa}(\gamma^2_{\beta\kappa})^{-1}][\gamma^1_{\alpha\kappa}(\gamma^2_{\alpha\kappa})^{-1}]^{-1}[\gamma^1_{\alpha\beta}(\gamma^2_{\alpha\beta})^{-1}]
	\end{equation} and, by using the same argument of the last proposition proof (based on $\bar{\lambda}$), it is readily verified that $\partial \delta(\alpha,\beta,\kappa)=e_{\Gamma^1}$. 
\end{proof}	In this vein, the map expressed in Eq. (\ref{spc}) determines an element of $\check{H}^1(\mathcal{M},D)$. It is frequently called the {\it difference class} for the two $\Gamma-$structures and we shall denote it by $[\delta]_{(P_1,P_2)}$. The next proposition helps us understand some important properties of it. 
\begin{prop} 
Let $P_{\Gamma^i}$ $(i=1,2,3)$ be $\Gamma-$structures for $P_G$ ($P_i$ for short\footnote{Observe that, despite the notations, for all interesting cases, the Gamma's provide the same structure group.}). The difference class is such that $1)$ $[\delta]_{(P_1,P_2)}=[\delta]^{-1}_{(P_2,P_1)}$ and $2)$ $[\delta]_{(P_1,P_2)}\cdot [\delta]_{(P_2,P_3)}=[\delta]_{(P_1,P_3)}$.
\end{prop}
\begin{proof}
	From the very definition of the difference class, we have 
	\begin{enumerate}
\item $[\delta]_{(P_1,P_2)}=\gamma^1_{\alpha\beta}(\gamma^2_{\alpha\beta})^{-1}=(\gamma^2_{\alpha\beta}(\gamma^1_{\alpha\beta})^{-1})^{-1}=[\delta]^{-1}_{(P_2,P_1)}$;  
	
\item $[\delta]_{(P_1,P_2)}\cdot [\delta]_{(P_2,P_3)}=(\gamma^1_{\alpha\beta}(\gamma^2_{\alpha\beta})^{-1})\cdot(\gamma^2_{\alpha\beta}(\gamma^3_{\alpha\beta})^{-1})=\gamma^1_{\alpha\beta}(\gamma^3_{\alpha\beta})^{-1}=[\delta]_{(P_1,P_3)}$.
\end{enumerate}
\end{proof}
There is another relevant proposition before we state a powerful theorem. 
\begin{prop}
	Let $P_{i}$ $(i=1,2,3)$ be $\Gamma-$structures for $P_G$. There is an isomorphism $P_{\Gamma^2}\simeq P_{\Gamma^3}$ if, and only if, $[\delta]_{(P_2,P_3)}=e$.
\end{prop}	
\begin{proof}
	($\Rightarrow$) Consider $P_{\Gamma^2}\simeq P_{\Gamma^3}$. Then there exists $k:\Gamma_3\rightarrow \Gamma_2$ such that $\gamma^3_{\alpha\beta}\mapsto \gamma^2_{\alpha\beta}=k\gamma^3_{\alpha\beta}$. Therefore $[\delta]_{(P_2,P_3)}=\gamma^2_{\alpha\beta}(\gamma^3_{\alpha\beta})^{-1}=k\gamma^3_{\alpha\beta}(\gamma^3_{\alpha\beta})^{-1}=e$.
	
	($\Leftarrow$) If $[\delta]_{(P_2,P_3)}=e$, then the previous proposition property $2)$ collapses to $[\delta]_{(P_1,P_2)}=[\delta]_{(P_1,P_3)}$, that is, $\gamma^1_{\alpha\beta}(\gamma^2_{\alpha\beta})^{-1}=\gamma^1_{\alpha\beta}(\gamma^3_{\alpha\beta})^{-1}$. Entering with $(\gamma^1_{\alpha\beta})^{-1}$ from the left one arrives at $\gamma^2_{\alpha\beta}(x)=\gamma^3_{\alpha\beta}(x)$ for $x\in U_\alpha\cap U_\beta$.  
\end{proof}
\begin{theorem}
	Let $P$ and $P_{i}$ $(i=1,2)$ be $\Gamma-$structures for $P_G$. Then $P_1\simeq P_2$ if, and only if, $[\delta]_{(P,P_1)}=[\delta]_{(P,P_2)}$. 
\end{theorem}	
\begin{proof}
	Bearing in mind the last proposition, it suffices to demonstrate that $P_1\simeq P_2$ if, and only if, $[\delta]_{(P_1,P_2)}=e$. 

($\Rightarrow$) Starting with $\phi:P_1\rightarrow P_2$ being a $\Gamma-$structures isomorphism and setting by $\sigma^1_\alpha:U_\alpha\rightarrow P_1$ a system of local sections for $P_1$, then $\sigma^2_\alpha=\phi\sigma^1_\alpha$ gives a system of local sections for $P_2$ (see the diagram below). 
\begin{equation}
\xymatrix{
	\mathcal{M}\supset U_\alpha \ar[rdd]_{\sigma^2_\alpha} \ar[rr]^{\sigma^1_\alpha} & & P_1 \ar[ldd]^{\phi}\nonumber \\
	& & \\
	& P_2
}
\end{equation} Thus $\sigma^i_\beta=\sigma^i_\alpha\gamma^i_{\alpha\beta}$ $(i=1,2)$. Now, particularize $i=1$ and act with $\phi$ from the left to get $\sigma^2_\beta=\sigma^2_\alpha\gamma^1_{\alpha\beta}$. Therefore $\gamma^2_{\alpha\beta}=\gamma^1_{\alpha\beta}$, leading to $[\delta]_{P_1,P_2}=\gamma^1_{\alpha\beta}(\gamma^2_{\alpha\beta})^{-1}=e$. 

($\Leftarrow$) This step of the proof is quite subtle. Assuming $[\delta]_{(P_1,P_2)}=e$, there must exist a $1-$cochain $q$ such that $\delta_{\alpha\beta}=q(\alpha)q^{-1}(\beta)$, for $\partial \delta(\alpha,\beta,\kappa)=\delta_{\beta\kappa}$. Hence, for $z\in W_\alpha\cap W_\beta$ we have $\delta^{-1}_{\alpha\kappa}\delta_{\alpha\beta}=[q(\beta)q^{-1}(\kappa)][q(\alpha)q^{-1}(\kappa)]^{-1}[q(\alpha)q^{-1}(\beta)]=e$. Therefore $\delta_{\alpha\beta}=\gamma^1_{\alpha\beta}\gamma^2_{\alpha\beta}=q(\alpha)q(\beta)$ and a straightforward manipulation leads to $q^{-1}(\alpha)\gamma^1_{\alpha\beta}=q^{-1}(\beta)\gamma^2_{\alpha\beta}$. Recall that $\gamma^i_{\alpha\beta}: U_\alpha\cap U_\beta\rightarrow \Gamma^i$, $q:U_\alpha\rightarrow D$ and $D\in C(\Gamma^i)$. Thus, $q$ and $\gamma_{\alpha\beta}$ commute yielding 
\begin{equation}
\gamma^2_{\alpha\beta}=q^{-1}(\alpha)\gamma^1_{\alpha\beta}q(\beta). \label{pdp}
\end{equation} Now, recall that the fibre bundle structure allows for the existence of invertible projections such that $\pi^{-1}:U_\alpha\subset \mathcal{M}\rightarrow W_\alpha\subset T\mathcal{M}$. Besides, take $z\in W_\alpha$ and consider $\gamma_\alpha(z)\in \Gamma^1$ as the (unique\footnote{In fact, assume the existence of $\gamma'_\alpha$ such that $z=\sigma^1_\alpha(\pi_1(z))\gamma'_\alpha(z)$. Equaling this last expression to (\ref{visv}) and acting with $(\sigma_\alpha^1)^{-1}$ from the left one arrives at the uniqueness of $\gamma_\alpha$.}) element satisfying    
\begin{equation}\label{visv}
z=\sigma^1_\alpha(\pi_1(z))\gamma_\alpha(z),
\end{equation} where, obviously, $\pi_1^{-1}$ makes reference to $P_1$. Notice that in (\ref{visv}), we have a different order to the cross-section and the projection, so this composition in this order may not be the identity, and therefore, an adjustment coming from $\Gamma^1$ may be necessary. In the sequel, define a mapping $\phi:W_\alpha\rightarrow P_2$ by 
\begin{equation}
\phi_\alpha(z)=\sigma_\alpha^2(\pi_1(z))\cdot q(\alpha)\gamma_\alpha(z).\label{mai}
\end{equation} The idea is to prove that this local mapping defines a global one: the same map in every open set covering $\mathcal{M}$. Indeed, notice that $\phi_\beta(z)=\sigma_\beta^2(\pi_1(z))\cdot q(\beta)\gamma_\beta(z)$ and, in view of (\ref{pdp}), the usual cross sections relation $\sigma^2_\beta=\sigma^2_\alpha\gamma^2_{\alpha\beta}$ reads $\sigma^2_\beta=\sigma^2_\alpha q^{-1}(\alpha)\gamma^{1}_{\alpha\beta}q(\beta)$, by means of which we have 
\begin{equation}
\phi_\beta(z)=\sigma_\alpha^2(\pi_1(z)) q^{-1}(\alpha)\gamma^{1}_{\alpha\beta}q(\beta)\cdot q(\beta)\gamma_\beta(z).\label{com}
\end{equation} Nevertheless $q(\alpha)$ is a cochain, i. e., $q(\alpha)\in D\simeq \mathbb{Z}_2$ and thus $q^2=1$ $(q=q^{-1})$, simplifying (\ref{com}) to 
\begin{equation}
\phi_\beta(z)=\sigma_\alpha^2(\pi_1(z)) q(\alpha)\gamma^{1}_{\alpha\beta}\gamma_\beta(z).\label{bina}
\end{equation}

Consider $z\in W_\alpha\cap W_\beta$ and from (\ref{visv}) for $\alpha\mapsto\beta$ we have $[\sigma_\beta^1(\pi_1(z))]^{-1}z=\gamma_\beta(z)$ or, entering with the transition function from the left, $\gamma^1_{\alpha\beta}[\sigma_\beta^1(\pi_1(z))]^{-1}z=\gamma^1_{\alpha\beta}\gamma_\beta(z)$. The left-hand side of this last expression reads 
\begin{eqnarray}
[\sigma_\beta^1(\pi(z))(\gamma^1_{\alpha\beta})^{-1}]^{-1}z=[\sigma_\beta^1(\pi(z))\gamma^1_{\beta\alpha}]^{-1}z=[\sigma_\alpha^1(\pi(z))]^{-1}z 
\end{eqnarray} and using (\ref{visv}) we have $[\sigma_\alpha^1(\pi(z))]^{-1}z=\gamma_\alpha(z)$. Therefore $\gamma^1_{\alpha\beta}\gamma_\beta(z)=\gamma_\alpha(z)$ and back to (\ref{bina}) we are left with 
\begin{eqnarray}
\phi_\beta(z)=\sigma_\alpha^2(\pi_1(z)) q(\alpha)\gamma_\alpha(z)=\phi_\alpha(z), \hspace{.3cm} z\in W_\alpha\cap W_\beta. \label{bina2}
\end{eqnarray} Thus we have a global map $\phi:P_1\rightarrow P_2$. 

At the local sections level, this means the existence of $\kappa$ such that $\kappa(\sigma^1_\alpha)=\kappa(\sigma^1_\beta \gamma^1_{\beta\alpha})$ or, bearing in mind the discussion around Eq. (\ref{pdp}), $q(\alpha)\kappa\sigma^1_\alpha=q(\beta)\kappa\sigma^1_\beta\gamma^2_{\beta\alpha}$. In this vein, we can write 
\begin{eqnarray}
\sigma^2_\alpha=q(\alpha)\kappa\sigma_\alpha^1, \hspace{.3cm} \forall \alpha.\label{34}
\end{eqnarray} Now consider the expression (\ref{visv}). Taking it for $z\mapsto z\cdot \gamma$, where $\gamma\in \Gamma$, we have $z\cdot\gamma=\sigma^1_\alpha(\pi_1(z\cdot \gamma))\gamma_\alpha(z\cdot \gamma)$. On the other hand, by simple contraction $z\cdot\gamma=\sigma^1_\alpha(\pi_1(z))\gamma_\alpha(z)\cdot\gamma$ and hence 
\begin{eqnarray}
\sigma^1_\alpha(\pi_1(z\cdot \gamma))\gamma_\alpha(z\cdot \gamma)=\sigma^1_\alpha(\pi_1(z))\gamma_\alpha(z)\cdot\gamma.
\end{eqnarray} Inserting $q(\alpha)\kappa$ from the left and using Eq. (\ref{34}) we have 
\begin{equation}
\sigma^2_\alpha(\pi_1(z\cdot\gamma))\gamma_\alpha(z\cdot\gamma)=\sigma^2_\alpha(\pi_1(z))\gamma_\alpha(z)\cdot\gamma, \label{36}
\end{equation} from which the insertion of $q(\alpha)\in D\simeq \mathbb{Z}_2$ shows that 
\begin{equation}
\phi(z\cdot\gamma)=\phi(z)\cdot\gamma
\end{equation} and the isomorphism is complete. Besides, from the diagram below, we see that $\eta_1=\phi\eta_2$. 
\begin{equation}
\xymatrix{
	P_1 \ar[rdd]_{\phi} \ar[rr]^{\eta_1} & & P \ar[ldd]^{\eta_2^{-1}}\nonumber \\
	& & \\
	& P_2
}
\end{equation}
\end{proof}	

The demonstrated result shall be complemented by another, making correspondence between elements of the first \v{C}ech cohomology group and nonequivalent $\Gamma-$structures. Let us delve into that before ending this section.  
\begin{theorem}
	Let $(P_1,\eta_1)$ be a $\Gamma-$structure for $P$ and consider $[\delta]\in\check{H}^1(\mathcal{M},D)$. Then there exists one $\Gamma-$structure $(P_2,\eta_2)$ such that the difference class is $[\delta]$. 
\end{theorem}	 
\begin{proof}
Let $d$ be a cocycle with values in $D$ representing $[\delta]$. Call $\gamma^1_{\alpha\beta}(x)$ the transition functions for $P_1$ and define for $x\in U_\alpha\cap U_\beta$:
\begin{equation}\label{dd}
\gamma^2_{\alpha\beta}(x)=d^{-1}_{\alpha\beta}(x)\gamma^1_{\alpha\beta}(x). 
\end{equation} As a parenthetical remark, notice the appearance $d^{-1}$ in (\ref{dd}). In fact $d^{-1}_{\alpha\beta}:D\subset \Gamma\rightarrow  U_\alpha\cap U_\beta$, and (\ref{dd}) is well settled. Back to the main argumentation, the consistency condition is indeed satisfied by $\gamma^2_{\alpha\beta}$, for
\begin{equation}
\gamma^2_{\beta\kappa}(\gamma^2_{\alpha\kappa})^{-1}\gamma^2_{\alpha\beta}=d^{-1}_{\beta\kappa}\gamma^1_{\beta\kappa}(\gamma^1_{\alpha\kappa})^{-1}d_{\alpha\kappa}d^{-1}_{\alpha\beta}\gamma^1_{\alpha\beta}.  
\end{equation} Since $d\in D\subset C(\Gamma)$ the above expression may be recast 
\begin{equation}
  \gamma^2_{\beta\kappa}(\gamma^2_{\alpha\kappa})^{-1}\gamma^2_{\alpha\beta}=[d^{-1}_{\beta\kappa}d_{\alpha\kappa}d^{-1}_{\alpha\beta}]\gamma^1_{\beta\kappa}(\gamma^1_{\alpha\kappa})^{-1}\gamma^1_{\alpha\beta}.     	
\end{equation} The $\gamma^1$ part of the above expression amounts out to $e$ since $\gamma^1$ are transition functions. Moreover, $d$ is a cocicle and then $\partial d(\alpha,\beta,\kappa)=d_{\beta\kappa}d^{-1}_{\alpha\kappa}d_{\alpha\beta}=[d^{-1}_{\beta\kappa}d_{\alpha\kappa}d^{-1}_{\alpha\beta}]^{-1}=e$. Thus $\gamma^2_{\beta\kappa}(\gamma^2_{\alpha\kappa})^{-1}\gamma^2_{\alpha\beta}=e$. Therefore, there exists a principal bundle $(P_2,\pi,\mathcal{M},\Gamma)$ with $\gamma^2_{\alpha\beta}$ as transition functions. This establishes the $P_2$ fibre bundle existence. In order to show that $P_2$ is a $\Gamma-$structure for $P$, observe that 
\begin{equation}
\rho(\gamma^2_{\alpha\beta})=\rho(d^{-1}_{\alpha\beta}\gamma_{\alpha\beta})=\rho(d^{-1}_{\alpha\beta})\rho(\gamma^1_{\alpha\beta})=\rho(\gamma^1_{\alpha\beta})=g_{\alpha\beta}, 
\end{equation} where use was made of the fact that $\rho$ is a homomorphism and $d$ belongs to its kernel. The link with a $\Gamma-$structure is now guaranteed by Lemma II.1. Note that from the very existence of $[\delta]\in\check{H}^1(\mathcal{M},D)$, uniqueness follows from Theorem 2.\end{proof}	

The results of this section may be summed up by the following assertions: $\Gamma-$structures are allowed over a given manifold if, and only if, the second Stiefel-Whitney class $w_2$ is trivial. Moreover, in being $w_2$ trivial, there are nonequivalent $\Gamma-$structures if, and only if, $\check{H}^1(\mathcal{M},D)$ is nontrivial. Besides, there will be as many nonequivalent $\Gamma-$structures as there are elements in $\check{H}^1(\mathcal{M},D)$. Spinors coming from nonequivalent $\Gamma-$structures are called exotic. 

In order to have a flavor of the relation between $\check{H}^1$ nontriviality and nontrivial topology, suppose $S^1\subset \mathcal{M}$ as the only aspect of nontriviality for the base manifold topology\footnote{This case was applied, for instance, in using exotic spinors in modeling certain aspects of superconductivity \cite{petry}.}. Also, take (as usual to spinor bundles) $D\simeq\mathbb{Z}_2$. Hence the fundamental homotopy group of $\mathcal{M}$ will be \cite{nak} $\pi_1(\mathcal{M})=\pi_1(S^1)=\mathbb{Z}\supset n$. Any element $n$ of $\pi_1(S^1)$ may be classified according to its parity, producing a homomorphism between this fundamental group and $\mathbb{Z}_2$. On the other hand, elements from $\mathcal{M}$ taking values in $\mathbb{Z}_2$ are precisely the ones of $\check{H}^1(\mathcal{M},\mathbb{Z}_2)$. For this review, it analysis suffices to visualize the interplay between the first \v{C}ech cohomology group and the nontrivial topology of the base manifold\footnote{The formal relation may be evinced under the auspices of covering spaces: $\check{H}^1(\mathcal{M},\mathbb{Z}_2)$ classify the base manifold double covers. In contrast, homomorphisms from $\pi_1(\mathcal{M})$ to $\mathbb{Z}_2$ correspond to double covers of $\mathcal{M}$ \cite{mil}.}. Anyway, it can be asserted that $\check{H}^1(\mathcal{M},D)\simeq \hom(\pi_1(\mathcal{M})\to D)$ revealing that, as far as the base manifold is not simply connected, exotic spinors are in order. 

\section{Foundations - Part II: Diffeomorphisms}

Diffeomorphisms have profoundly impacted contemporary physics, from Lagrangian formalism to curved spacetimes. Understanding its impact on different representation fields is of primary importance. Concerning this presentation, it is illustrative to know how diffeomorphisms impact spinors. More than an exercise of completeness, the analysis has an impact, for instance, on the vacuum-generating functional construction for spinor field quantum theory \cite{AI}. We shall construct this section following \cite{dape}. 

\subsection{Spinors as `scalars' under diffeomorphisms}

Let $\mathcal{Diff}(\mathcal{M})$ denotes the base manifold group of diffeomorphisms\footnote{All along this section we are assuming that transformations of $\mathcal{Diff}(\mathcal{M})$ preserve orientability of the $n$ dimensional base manifold. Besides, for simplicity, we shall deal with a Riemannian base manifold. Parallelization issues apart, generalizations to pseudo-Riemannian manifolds are somewhat direct.}. Suppose the existence of a principal bundle $P_{GL^+(n)}$ (with $GL^+(n)$ as structure group) over $\mathcal{M}$ endowed with a principal subbundle $P_{SO(n)}\subset P_{GL^+(n)}$. Similarly, take $P_{Spin(n)}\subset P_{\widetilde{GL}^+(n)}$ where $\widetilde{GL}^+(n)$ is the double cover of $GL^+(n)$. Hence $P_{Spin(n)}$ ($P_{\widetilde{GL}^+(n)}$) shall serve as $\Gamma-$structure to $P_{SO(n)}$ ($P_{GL^+(n)}$). More precisely, there exists $\rho:\widetilde{GL}^+(n)\rightarrow GL^+(n)$ such that $\ker(\rho)=\mathbb{Z}_2\subset C(\widetilde{GL}^+(n))$ and a map $\eta:P_{\widetilde{GL}^+(n)}\rightarrow P_{GL^+(n)}$ allowing, somewhat naturally, for $\eta|_{P_{Spin(n)}}:P_{Spin(n)}\rightarrow P_{SO(n)}$ and leading $(P_{Spin(n)},\eta|_{P_{Spin(n)}})$ to serve as a spin structure. The situation is, then, as follows: for a given spacetime metric, suppose the existence of a spinor structure. Thus, it is possible to construct an associated fibre bundle
\begin{equation}
\tilde{\Omega}\equiv P_{Spin(n)}\times_{[Spin(n)]}\widetilde{GL}^+(n).
\end{equation} Now, define
\begin{eqnarray}
\eta&:&\left.\tilde{\Omega}\rightarrow \Omega \right.\nonumber\\&&\left. (u,a)\mapsto \eta(u,a)=[\eta|_{P_{Spin(n)}}(u),\rho(a)],\right.
\end{eqnarray} so that the identification $\Omega=P_{SO(n)}\times_{[SO(n)]}GL^+(n)$ is only natural. The existence condition of the structure $(\tilde{\Omega},\eta)$ rests upon the usual topological obstruction, i.e., the vanishing of the second Stiefel-Whitney characteristic class. 

To properly investigate the impact of diffeomorphisms on spinor fields, take $f \in \mathcal{Diff}(\mathcal{M})$. Taking the pullback into account, $g'=f^*g$, $f$ is naturally lifted to $\wp \in Aut(GL^+(n))$. Restricting the action of $\wp$ to $SO(n)$, it maps orthonormal frames of $g'$ to orthonormal frames of $g$:
\begin{equation}
\wp: P_{SO(n)}|_{g'}\rightarrow P_{SO(n)}|_{g}. 
\end{equation} As the previous construction asserts (denoting the metric explicitly)  $\eta|_{P_{Spin(n)}|_{g'}}:P_{Spin(n)}|_{g'}\rightarrow P_{SO(n)}|_{g'}$ and, of course, $\eta|_{P_{Spin(n)}|_{g}}:P_{Spin(n)}|_{g}\rightarrow P_{SO(n)}|_{g}$. Presumably, it may be possible to construct, out from 
\begin{equation}
\{\wp, \eta|_{P_{Spin(n)}|_{g'}}, \eta|_{P_{Spin(n)}|_{g}}\},
\end{equation} a unique (up to trivial isomorphisms) induced isomorphism $\tilde{\wp}:P_{Spin(n)}|_{g'}\rightarrow P_{Spin(n)}|_g$ by 
\begin{eqnarray}
\tilde{\wp}=\eta^{-1}|_g\cdot \wp \cdot \eta|_{g'}, \label{ind}
\end{eqnarray} where $\eta|_g$ stands for $\eta|_{P_{Spin(n)}|_{g}}$ and similarly for $\eta|_{g'}$. If Eq. (\ref{ind}) can be induced, the following diagram can also be drawn:
\begin{equation}\label{aha}
 \xymatrixcolsep{4pc}\xymatrix{
 	P_{Spin(n)}|_{g'} \ar[d]_-{\eta|_{g'}} 
 	\ar[r]^-{\tilde{\wp}} & P_{Spin(n)}|_g\\
 	P_{SO(n)}|_{g'}\ar[r]^-{\wp}& P_{SO(n)}|_{g} \ar[u]_-{\eta^{-1}|_g}
}
\end{equation} Notice that $\tilde{\wp}\in Aut(Spin(n))$ is still in the spinor structure $P_{Spin(n)}|_g$, despite the action of the diffeomorphism. 

To appreciate the consequence of $\tilde{\wp}$ induced by Eq. (\ref{ind}), take $\sigma=\{(1/2,0),(0,1/2),(1/2,0)\oplus(0,1/2)\}$ denoting spin $1/2$ representations of $Spin(n)$. Observe that for a metric $g$ a spinor is $P_{Spin(n)}|_g\times_\sigma \mathbb{C}^4$. In the sequel, consider an open set $U\subset \mathcal{M}$ and\footnote{Naturally, for $g'=f^* g$, $U'=f^{-1}(U)$.} $\{r\}$ a set (a field) of orthogonal frames for $g$ in $U$ such that $\{r\}\supset r: U\rightarrow P_{SO(n)}|_g$. In an analog fashion, consider also $\{\tilde{r}\}\supset \tilde{r}:U\rightarrow P_{Spin(n)}|_g$, such that $\eta|_g\cdot \tilde{r}=r$. With those ingredients, we can set the relation between orthogonal frames for $g$ with their counterparts for $g'$ by 
\begin{eqnarray}
r'=\wp^{-1}\cdot r\cdot f, \;\;\;\;\; \text{and} \;\;\;\;\;
\tilde{r}'=\tilde{\wp}^{-1}\cdot \tilde{r} \cdot f.
\end{eqnarray} A local spinor field needs only one more piece to be constructed: a suitable map $\xi:P_{Spin(n)}|_g\rightarrow \mathbb{C}^4$ allowing for the definition of a local section of $P_{Spin(n)}|_g\times_\sigma \mathbb{C}^4$ by
\begin{eqnarray}\label{spi}
\psi&:&\left.U\subset \mathcal{M}\rightarrow \mathbb{C}^4 \right.\nonumber\\&&\left. x\mapsto \psi(x)=\xi(\tilde{r}(x)).\right.
\end{eqnarray} The induction of $\tilde{\wp}$ as in Eq. (\ref{ind}) allows for an equivariance in defining local spinor field transformation: $\xi'=\xi\circ \tilde{\wp}$. Therefore, $\psi'(x):=\xi'(\tilde{r}'(x))$ has such a behavior that, although the frame is well defined for $g'$, $\psi'(x)$ is still a section of $P_{Spin(n)}|_g\times_\sigma \mathbb{C}^4$. Explicitly,
\begin{eqnarray}
\psi'=\xi\cdot\underbrace{\tilde{\wp}\cdot\tilde{\wp}^{-1}}_{Id_{P_{Spin(n)}|_g}}\cdot\,\tilde{r}\cdot f, 
\end{eqnarray} or, making explicit the argument
\begin{eqnarray}
\psi'(x)=\xi(\tilde{r}(f(x))),
\end{eqnarray} which, by means of (\ref{spi}), reads simply $\psi'(x)=\psi(f(x))$, confirming the idea of spinor as ``scalar'' objects under diffeomorphisms. 

It is important to remark the induced fibre bundle map $\tilde{\wp}$ centrality to the above argument. If an obstruction to (\ref{ind}) is on order, then a given diffeomorphism transformation could affect spinors non-trivially. This central aspect is investigated in what follows. 

\subsection{Obstructions to the trivial behavior of spinors under diffeomorphisms}       
	
We shall define additional tools and twist the notation a little, by convenience, to properly approach the alluded possible obstructions. Let be $\{U_\alpha\}$ a covering of $\mathcal{M}$ and construct bundle charts 
\begin{eqnarray}\label{car}
\tilde{\varphi}_\alpha: U_\alpha\times\widetilde{GL}^+(n)\rightarrow \tilde{\Omega}|_{U_\alpha}
\end{eqnarray} and 
\begin{eqnarray}\label{ta} 
\varphi_\alpha: U_\alpha\times GL^+(n)\rightarrow \Omega|_{U_\alpha}.
\end{eqnarray} Besides, denote the corresponding transition functions by $\tilde{\gamma}_{\alpha\beta}:U_\alpha\cap U_\beta\rightarrow \widetilde{GL}^+(n)$ and $\gamma_{\alpha\beta}:U_\alpha\cap U_\beta\rightarrow GL^+(n)$, respectively. Since $\rho:\widetilde{GL}^+(n)\rightarrow GL^+(n)$ ($\gamma_{\alpha\beta}=\rho\cdot\tilde{\gamma}_{\alpha\beta}$), there exists $\eta$ such that
\begin{equation}
\eta:\tilde{\varphi}_\alpha(x,a)\mapsto \varphi_\alpha(x,\rho(a)),\label{delta}
\end{equation} for $x\in U_\alpha\subset \mathcal{M}$ and $a\in\widetilde{GL}^+(n)$. 

Now, in the context of $f\in \mathcal{Diff}(\mathcal{M})$, consider a covering $\{U'_\alpha\}$ with $U'_\alpha=f^{-1}(U_\alpha)$. It is, then, possible to write $\varphi':U'_\alpha\times GL^+(n)\rightarrow \Omega|_{U_\alpha}$ or 
\begin{equation}
\varphi':f^{-1}(U_\alpha)\times GL^+(n)\rightarrow \Omega|_{U_\alpha},
\end{equation} from which the following recognition is due 
\begin{equation}
\varphi'_\alpha(f^{-1}(x),a)=\wp^{-1}\varphi_\alpha(x,a) \label{2stars}
\end{equation} and $\gamma'_{\alpha\beta}(f^{-1}(x))=\gamma_{\alpha\beta}(x)$, as well. Analogously to the construction leading to (\ref{delta}), it is quite conceivable to have  ($\forall\, x\in U_\alpha$) 
\begin{equation}
\eta':\tilde{\varphi}'_\alpha(f^{-1}(x),a)\mapsto \varphi'_\alpha(f^{-1}(x),\rho(a)) \label{1star}
\end{equation} and $\gamma_{\alpha\beta}'=\rho\circ\tilde{\gamma}'_{\alpha\beta}$. 

As a last remark before to stay an important result, consider a suitable version of (\ref{spc}) given by
\begin{eqnarray}
\delta_{\alpha\beta}=\tilde{\gamma}'_{\alpha\beta}(f^{-1}(x))\tilde{\gamma}_{\alpha\beta}^{-1}(x), \hspace{.2cm} x\in U_\alpha\cap U_\beta.\label{nd}
\end{eqnarray} It can be readily verified that $\delta_{\alpha\beta}\in \mathbb{Z}_2 \subset \widetilde{GL}^+(n)$, for 
$\rho\cdot \delta_{\alpha\beta}=\rho\cdot\tilde{\gamma}'_{\alpha\beta}(f^{-1}(x))\cdot\tilde{\gamma}_{\alpha\beta}^{-1}(x)=\gamma'_{\alpha\beta}(f^{-1}(x))\cdot\tilde{\gamma}_{\alpha\beta}^{-1}(x)$ and using $\gamma'_{\alpha\beta}(f^{-1}(x))=\gamma_{\alpha\beta}(x)=\rho\cdot \tilde{\gamma}_{\alpha\beta}$, we are left with 
\begin{eqnarray}
\rho\cdot \delta_{\alpha\beta}=\rho\cdot \tilde{\gamma}_{\alpha\beta}(x)\cdot\tilde{\gamma}_{\alpha\beta}^{-1}(x)=\rho\cdot Id_{\widetilde{GL}^+(n)}=Id_{GL^+(n)}.
\end{eqnarray} Therefore $\delta_{\alpha\beta}\in \ker(\rho)\simeq \mathbb{Z}_2$. 
\begin{prop}
	$\delta_{\alpha\beta}$ is a cocycle. 
\end{prop}	  
\begin{proof}
	Take $x\in U_\alpha\cap U_\beta\cap U_\kappa$. The cobordism $\partial \delta(\alpha,\beta,\kappa)=\delta_{\beta\kappa}\delta^{-1}_{\alpha\kappa}\delta_{\alpha\beta}$ reads, form Eq. (\ref{nd}),
\begin{eqnarray}
\partial\delta(\alpha,\beta,\kappa)=[\tilde{\gamma}'_{\beta\kappa}\tilde{\gamma}_{\beta\kappa}][\tilde{\gamma}'_{\alpha\kappa}\tilde{\gamma}_{\alpha\kappa}]^{-1}[\tilde{\gamma}'_{\alpha\beta}\tilde{\gamma}_{\alpha\beta}].
\end{eqnarray} Define, as before, an homomorphism restricted to the kernel $\bar{\lambda}:\widetilde{GL}^+(n)\rightarrow\widetilde{GL}'^+(n)$, such that $\tilde{\gamma}_{\alpha\beta}'=\bar{\lambda}\tilde{\gamma}_{\alpha\beta}$. Thus, it can be readily obtained $\partial\delta(\alpha,\beta,\kappa)=e$.   
\end{proof}	We see, therefore, that $\delta_{\alpha\beta}$ as defined in Eq. (\ref{nd}) gives rise to an element $[\delta]\in \check{H}^1(\mathcal{M},\mathbb{Z}_2)$. Now, we can evince a central theorem on the trivial behavior obstruction of spinor fields under diffeomorphisms.  
\begin{theorem} 
 Let $f\in \mathcal{Diff}(\mathcal{M})$ and consider two principal fibre bundles, $P$ and $P'$, with their double covering morphisms given by $\eta$ and $\eta'$, respectively. The automorphism $\wp$ of $f$ lifts to an isomorphism $\tilde{\wp}:P'\rightarrow P$ with $\eta\cdot\tilde{\wp}=\wp\cdot\eta'$ if, and only if, $\check{H}^1(\mathcal{M},\mathbb{Z}_2)$ is trivial. 
\end{theorem}
\begin{proof}
($\Rightarrow$) If $[\delta]$ is trivial, then there exists a $0-$cochain, say $\lambda_\alpha:U_\alpha\rightarrow\mathbb{Z}_2$, such that\footnote{Notice that $\partial\delta(\alpha,\beta,\kappa)=\partial^2\lambda(\alpha,\beta)=\lambda_\beta\lambda_\kappa^{-1}[\lambda_\alpha\lambda_\kappa^{-1}]^{-1}\lambda_\alpha\lambda_\beta^{-1}=e$, with effect.} $\delta_{\alpha\beta}=\partial\lambda(\alpha,\beta)=\lambda_\beta\lambda_\alpha^{-1}$. Define the bundle chart mapping
\begin{eqnarray}
\tilde{\wp}&:&\left.\tilde{P}'\rightarrow \tilde{P}\right.\nonumber\\
&&\tilde{\varphi}'_{\alpha}(f^{-1}(x),a)\mapsto \tilde{\varphi}_{\alpha}(x,\lambda_\alpha(x)a).\label{3e}
\end{eqnarray} Applying $\wp\cdot\eta'$ upon $\tilde{\varphi}'_{\alpha}(f^{-1}(x),a)$ one get, using Eq. (\ref{1star}) and (\ref{2stars}), 
\begin{eqnarray}
(\wp\cdot\eta')\tilde{\varphi}'_\alpha(f^{-1}(x),a)=\wp\varphi'_\alpha(f^{-1}(x),\rho(a))=\varphi_\alpha(x,a).\label{ala}
\end{eqnarray} On the other hand, the application of $(\eta\cdot\tilde{\wp})$ on $\tilde{\varphi}'_{\alpha}(f^{-1}(x),a)$, using (\ref{3e}) and (\ref{delta}), yields 
\begin{eqnarray}
(\eta\cdot\tilde{\wp})\tilde{\varphi}'_{\alpha}(f^{-1}(x),a)=\eta\tilde{\varphi}_\alpha(x,\lambda_\alpha(x)a)=\varphi_\alpha(x,\rho(\lambda_\alpha(x)a)). \label{aala}
\end{eqnarray} Nevertheless, $\lambda_\alpha\in\mathbb{Z}_2\simeq \ker(\rho)$ and hence the central homomorphism in (\ref{aala}) reads $\rho(\lambda_\alpha(x)a)=\rho(\lambda_\alpha(x))\rho(a)=\rho(a)$. Then, (\ref{aala}) reads $(\eta\cdot\tilde{\wp})\tilde{\varphi}'_{\alpha}(f^{-1}(x),a)=\varphi_\alpha(x,a)$ which, when comparing with (\ref{ala}), gives 
\begin{eqnarray}
\eta\cdot\tilde{\wp}=\wp\cdot\eta', \hspace{.1cm} \forall \hspace{.1cm} \tilde{\varphi}'_{\alpha}.\label{ago}
\end{eqnarray}   

($\Leftarrow$) From  (\ref{3e}) it can be immediately seen that 
\begin{equation}
(\eta\cdot\tilde{\wp})\tilde{\varphi}'_\alpha(f^{-1}(x),a)=\eta\tilde{\varphi}_\alpha(x,\lambda_\alpha(x)a). \label{ou}
\end{equation} Using (\ref{ago}) in the left-hand side of (\ref{ou}), along with (\ref{delta}) and (\ref{1star}), we are left with 
\begin{equation}
\wp\varphi'_\alpha(f^{-1}(x),\rho(a))=\varphi_\alpha(x,\rho(a\lambda_\alpha(x))). \label{tra}
\end{equation} Equation (\ref{2stars}) enables us to reduce (\ref{tra}) to $\varphi_\alpha(x,\rho(a))=\varphi_\alpha(x,\rho(\lambda_\alpha(x)a))$ and, thus, one concludes that $\rho(\lambda_\alpha(x))=e$, that is $\lambda_\alpha\in\mathbb{Z}_2\simeq\ker(\rho)$. This fact, along with  $\gamma'_{\alpha\beta}(f^{-1}(x))=\rho\cdot \tilde{\gamma}'_{\alpha\beta}(f^{-1}(x))$, allows for 
\begin{eqnarray}
\gamma'_{\alpha\beta}(f^{-1}(x))\lambda_\alpha(x)=\rho\cdot [\tilde{\gamma}'_{\alpha\beta}(f^{-1}(x))\lambda_\alpha]=\rho\cdot\tilde{\gamma}'_{\alpha\beta}(f^{-1}(x)) 
\end{eqnarray} ans similarly
\begin{eqnarray}
\gamma'_{\alpha\beta}(f^{-1}(x))\lambda_\beta(x)=\rho\cdot [\tilde{\gamma}'_{\alpha\beta}(f^{-1}(x))\lambda_\beta]=\rho\cdot\tilde{\gamma}'_{\alpha\beta}(f^{-1}(x)), 
\end{eqnarray} from which we read $\gamma'_{\alpha\beta}(f^{-1}(x))\lambda_\beta(x)=\gamma'_{\alpha\beta}(f^{-1}(x))\lambda_\alpha(x)$. Using $\gamma'_{\alpha\beta}(f^{-1}(x))=\gamma_{\alpha\beta}(x)$ on the left-hand side of this previous equation and equating it, we have 
\begin{equation}
\lambda_\beta \lambda_\alpha^{-1}=\gamma'_{\alpha\beta}(f^{-1}(x))\gamma^{-1}_{\alpha\beta}(x). \label{lala}
\end{equation} Calling $\delta_{\alpha\beta}=\gamma'_{\alpha\beta}(f^{-1}(x))\gamma^{-1}_{\alpha\beta}(x)$ (as before) it follows straightforwardly that $\mathbb{Z}_2 \ni \delta_{\alpha\beta}=\partial\delta(\alpha,\beta)$ and, thus,  $\check{H}^1(\mathcal{M},\mathbb{Z}_2)$ is trivial. \end{proof}

The result encoded in this last theorem is relevant. It imposes a strong condition on the lifting presented in Eq. (\ref{ind}). As discussed before, this isomorphism is necessary (and sufficient) to warrant the spinor field as a scalar under diffeomorphisms. In the light of the proved theorem, we see this is indeed the case of diffeomorphisms homotopic to the identity. When the $\mathcal{M}$ topology is not trivial, however, (in the sense that $\pi_1(\mathcal{M})\neq 0$), this result cannot be applied to the entire base manifold. Recall that when $\check{H}^1(\mathcal{M},\mathbb{Z}_2)$ is not trivial, there are more than one nonequivalent spinor structures coexisting. Moreover, there are as many nonequivalent spinor fields as elements in $\check{H}^1(\mathcal{M},\mathbb{Z}_2)$. Since Eq. (\ref{ind}) is not achieved for nontrivial $\check{H}^1(\mathcal{M},\mathbb{Z}_2)$, there is no guarantee that the behavior of a spinor field under a given diffeomorphism preserves the spinor class (see discussion around (\ref{aha})). That is, diffeomorphisms that are not connected to the identity permutes among spin structures. 

The behavior of spinor fields under isometries was investigated in Refs. \cite{um,dois}, but the general diffeomorphism analysis allows understanding important physical consequences. In a rough description, consider a manifold endowed with a global system of tetrad fields $e_{a\mu}$ such that $g^{\mu\nu}e_{a\mu}e_{b\nu}=\eta_{ab}$, where $\eta_{ab}$ is the Minkowski metric, and the base manifold (Greek) and frame (Latin) indexes run from $0,1,2,3$. A transformation impacting the base manifold metric, say $\xi$, is usually absorbed into a spinor reparametrization, such that the generating functional is invariant (for details, see \cite{AI}). Nevertheless, when $\check{H}^1(\mathcal{M},\mathbb{Z}_2)$ is not trivial, and a permutation of spin structures is in order, some sophistication shall be incorporated into the analysis to ensure invariance. 
In this vein, an example of gravity and spin-1/2 field combined functional, which is entirely invariant under $\xi$, is given by 
\begin{eqnarray}
Z\supset \int[De]\Bigg\{\sum_{j\in \check{H}^1(\mathcal{M},\mathbb{Z}_2)}\int [D_j\bar{\psi}][D_j\psi]e^{i\int_{\mathcal{M}}\mathcal{L}(e\xi_j,\psi)}\Bigg\} e^{iG[e]},  
\end{eqnarray} where $G[e]$ stands for the gravitational action, and $\mathcal{L}$ the spin-1/2 field Lagrangian density in which the covariant derivative term encompasses nontrivial topology effects into a corrected spin connection (see next section). Thus, regardless of the spinor structure permutation effect due to the diffeomorphism, the sum over all $\check{H}^1(\mathcal{M},\mathbb{Z}_2)$ elements renders the term between brackets invariant.           
  
\section{Exotic spinors}

Suppose from now on a manifold $\mathcal{M}$ endowed with nontrivial topology, in the sense of our previous discussion, i.e., nontrivial $\check{H}^1(\mathcal{M},\mathbb{Z}_2)$, but with vanishing Stiefel-Whitney classes $w_k$ ($k=1,2,3,4$). We shall investigate here how the exotic spinor equation of motion is impacted. For the initial setup just described, it is licit to say that at least two nonequivalent spinorial structures exist. Besides, since we are interested in spinors themselves, it is useful to recall the following fact: from a topological spaces point of view, there is a homeomorphism between the fibre of the principal frame bundle $P_{Spin(1,3)}\times_\sigma \mathbb{C}^4$ and the structure group $L_+^\uparrow$, where $\sigma$ stands for $(0,1/2),(1/2,0)$, or $(0,1/2)\oplus(1/2,0)$ representations and $L_+^\uparrow$ denotes the orthochronous proper Lorentz subgroup. Formally, there is then a bundle map, say $\zeta$, such that the diagram bellow commutes:   
\begin{equation}
\xymatrix{
	(P,\pi,\mathcal{M},L_+^\uparrow) \ar[rdd]_{\pi} \ar[rr]^{\zeta} & & P_{Spin(1,3)}\times_\sigma \mathbb{C}^4 \ar[ldd]^{\pi'}\nonumber \\
	& & \\
	& \mathcal{M}
}
\end{equation} and with its own lifting structure such that, if $z\in P$ and $g\in Spin(1,3)$, then $\eta(zg)=\eta(z)\rho(g)$, as before. We shall refer to usual spinors as sections of $P_{Spin(1,3)}\times_\sigma \mathbb{C}^4$, i.e., $\psi\in \Omega(\mathcal{M}):= \sec P_{Spin(1,3)}\times_\sigma \mathbb{C}^4$. Quite analogously, the very same structure can be, {\it mutatis mutandis},  inputted to exotic spinors $\tilde{\psi}\in\tilde{\Omega}(\mathcal{M}):=\sec\tilde{P}_{Spin(1,3)}\times_\sigma \mathbb{C}^4$.  
 
Consider $x\in U_\alpha\cap U_\beta\subset \mathcal{M}$ and mappings $\gamma_{\alpha\beta}$ and $\tilde{\gamma}_{\alpha\beta}$ from this intersection to $Spin(1,3)$. According to our previous discussion (see Eq. (\ref{spc}) taking $\gamma^1_{\alpha\beta}=\tilde{\gamma}_{\alpha\beta}$ and $\gamma^2_{\alpha\beta}=\gamma_{\alpha\beta}$), these mappings are related by $\tilde{\gamma}_{\alpha\beta}=\gamma_{\alpha\beta}\delta_{\alpha\beta}$, with $\delta_{\alpha\beta}\in \mathbb{Z}_2$. Take a faithful, linear and invertible map $\Delta:Cl_{1,3}\rightarrow M(4,\mathbb{C})$, relevant here since $Spin(1,3)\simeq Cl^+_{1,3}\subset Cl_{1,3}$. The further analysis can be simplified by defining $U(1)$ functions $\xi_\alpha:U_\alpha\rightarrow \mathbb{C}$ in such a way that    
\begin{eqnarray}\label{aa}
\quad \qquad \qquad \qquad \qquad 
\xi_\alpha(x)=\Delta(\delta_{\alpha\beta}(x))\xi_\beta(x),\qquad \forall \; x\!\in U_\alpha\cap U_\beta. \label{s0}
\end{eqnarray} More comments about $\xi_\alpha$ functions are necessary and shall be made soon. By now, notice that $\delta_{\alpha\beta}\in \mathbb{Z}_2$ and, since $\Delta$ is faithful, $\Delta(\delta_{\alpha\beta})=\pm 1$. Therefore
\begin{eqnarray}\label{shall}
\frac{\xi_\alpha}{\xi_\beta}=\pm 1,
\end{eqnarray} a fact that make explicit the discontinuity of $\xi_\alpha(x)$ as some hypersurface $\Sigma\subset U_\alpha\cap U_\beta$ is crossed. Besides, notice that $\xi_\alpha$ encodes the topological nontriviality at the local sections level, so to speak. A given function exists if the group characterizing this topological nontriviality is such that any sequence of successive compositions does not trivialize the group element. That is to say, there is no $k\in\mathbb{N}^*$ such that $[\delta]^k=e$. This is guaranteed if $\check{H}^1(\mathcal{M},\mathbb{Z}_2)$ has no group torsion.

Back to the main point, by restricting sections to a given open set, i.e. $\Omega(\mathcal{M})|_{U_{\alpha}\cap U_\beta}$, and bearing in mind the role played by the transition functions, it is only natural to relate local spinors $\psi_\alpha\in \Omega(\mathcal{M})|_{U_{\alpha}\cap U_\beta}$ by $\psi_\alpha=\Delta(\gamma_{\alpha\beta})\psi_\beta$. In parallel, as mentioned, exotic spinors obey similar rules. In fact, for $\tilde{\psi}_\alpha\in\tilde{\Omega}(\mathcal{M})|_{U_\alpha\cap U_\beta}$ one has $\tilde{\psi}_\alpha=\Delta(\tilde{\gamma}_{\alpha\beta})\tilde{\psi}_\beta$ and, since $\Delta(\tilde{\gamma}_{\alpha\beta})=\Delta(\gamma_{\alpha\beta}\delta_{\alpha\beta})=\Delta(\gamma_{\alpha\beta})\Delta(\delta_{\alpha\beta})$, it follows that 
\begin{equation}
\tilde{\psi}_{\alpha}=\Delta(\gamma_{\alpha\beta})\Delta(\delta_{\alpha\beta})\tilde{\psi}_\beta.\label{s1}
\end{equation} The linearity of $\Delta$ can be further explored by acting upon Eq. (\ref{s0}) leading to $\Delta(\xi_\alpha)=\Delta(\delta_{\alpha\beta})\Delta(\xi_\beta)$. Therefore, from (\ref{s1}) 
\begin{equation}
\Delta(\xi_\alpha)\tilde{\psi}_\alpha=\Delta(\delta_{\alpha\beta})\Delta(\xi_\beta)\Delta(\gamma_{\alpha\beta})\Delta(\delta_{\alpha\beta})\tilde{\psi}_\beta.\label{sss}
\end{equation} Of course, $\Delta^2(\delta_{\alpha\beta})=1$. Hence, provided that the commutator $[\Delta(\gamma_{\alpha\beta}),\Delta(\xi_\beta)]$ vanishes, we are left with $\Delta(\xi_\alpha)\tilde{\psi_\alpha}=\Delta(\gamma_{\alpha\beta})\Delta(\xi_\beta)\tilde{\psi}_\beta$. Let us compare this expression with $\psi_\alpha=\Delta(\gamma_{\alpha\beta})\psi_\beta$. In general, from the very form of these relations, one can see that $\Delta(\xi_\alpha)$ entails a {\it local} map between the usual and exotic sections in the sense that $\Delta(\xi_\alpha)\tilde{\psi}_\alpha$ transforms as a local component of $\Omega(\mathcal{M})|_{U_\alpha}$. This is sufficient to directly relate $\Delta(\xi_\alpha)\tilde{\psi}_\alpha$ and $\psi_\alpha$, but here we shall argue on the validity of this relation via a heuristic argument. Consider a spinorial (Latin) index running from all the spinor entries for convenience. Thus, in components we have $(\psi_\alpha)_a=\Delta(\gamma_{\alpha\beta})_{aa'}(\psi_\beta)_{a'}$. Admitting $\Delta(\gamma_{\alpha\beta})$ diagonal and $(\psi_\beta)_{a'}$ non null\footnote{If this is not the case, we just have trivial null entries related by $(\psi_\alpha)_a=\Delta(\gamma_{\alpha\beta})_{aa'}(\psi_\beta)_{a'}$.}, it is possible to write $(\psi_\alpha)_{l}(\psi_\beta)_m^{-1}=\Delta(\gamma_{\alpha\beta})_{lm}$, from which we have
$\Delta(\xi_\alpha)_{ll'}(\tilde{\psi}_\alpha)_{l'}=(\psi_\alpha)_{l}(\psi_\beta)_m^{-1}\Delta(\xi_\beta)_{mn}(\tilde{\psi}_\beta)_n$. Multiplying this last expression by $(\psi_\alpha)_l^{-1}$ (admitting $(\psi_\alpha)_l$ non null)  
\begin{eqnarray}
(\psi_\alpha)_l^{-1}\Delta(\xi_\alpha)_{ll'}(\tilde{\psi}_\alpha)_{l'}=(\psi_\beta)_m^{-1}\Delta(\xi_\beta)_{mn}(\tilde{\psi}_\beta)_n.\label{h1}
\end{eqnarray} Since each side refers to a different open set, one must conclude that they must be equal to the same constant $C$. Hence $(\psi_\alpha)_l^{-1}\Delta(\xi_\alpha)_{ll'}(\tilde{\psi}_\alpha)_{l'}=C$. Besides, notice that $C$ is a vacuous constant (which could eventually be even absorbed in a spinor redefinition); it can be made equal to unity without losing generality and, therefore,  
\begin{equation}
\psi_\alpha=\Delta(\xi_\alpha)\tilde{\psi}_\alpha.\label{h2}
\end{equation}
 
Once Eq. (\ref{h2}) refers to a unique open set, giving up the open set index is common, assuming its validity to every local section. It is indeed true, but we must refrain from doing so at this point since, as mentioned, $\xi_\alpha$ has an intricate behavior. Let us, then, keep a reference to the open section for a while. Using Eq. (\ref{h2}), the bilinear covariants relation between usual and exotic spinors can be computed. The Dirac duals are related by $\bar{\psi}_\alpha=\bar{\tilde{\psi}}_\alpha\gamma^0\Delta^\dagger(\xi_\alpha)\gamma^0$. Let us denote by $\Gamma$ any matrix of the set composed by a base for $M(4,\mathbb{C})$ in terms of Dirac matrices. It is straightforward to see that 
\begin{equation}\label{h3}
\bar{\psi}_\alpha\Gamma \psi_\alpha=\bar{\tilde{\psi}}_\alpha\gamma^0\Delta^\dagger(\xi_\alpha)\gamma^0\Gamma\Delta(\xi_\alpha)\tilde{\psi}_\alpha,
\end{equation} and if $\Delta(\xi_\alpha)$ commutes with $\Gamma$, Eq. (\ref{h3}) reduces to $\bar{\psi}_\alpha\Gamma \psi_\alpha=\bar{\tilde{\psi}}_\alpha\Delta^\dagger(\xi_\alpha)\Delta(\xi_\alpha)\Gamma\tilde{\psi}_\alpha$. In this case, one sees that $\Delta(\xi_\alpha)$ unitarity enables bilinear covariants invariance. The condition $[\Delta(\xi_\alpha),\Gamma]=0$ is, however, quite stringent. In fact, there is no much freedom in the $M(4,\mathbb{C})$ part of $\Delta(\xi_\alpha)$ and it must be proportional to the identity $\mathbb{1}_{M(4,\mathbb{C})}$ \cite{tese}. Notice, in particular, that it satisfies the previous necessary constraint $[\Delta(\gamma_{\alpha\beta}),\Delta(\xi_\beta)]=0$. Since the matrix part of $\Delta(\xi_\alpha)$ is already settled as trivial, we shall pay attention to its unitarity and open set dependence in the sequel. Before, however, we note that the collected information is enough to ensure the following property: let $O\in M(4,\mathbb{C})$ be an operator acting in $\Omega(\mathcal{M})|_{U_\alpha}$ engendering an eigenspinor relation, that is $O\psi_\alpha=\varepsilon\psi_\alpha$. Assuming $\Delta(\xi_\alpha)$ linear, we have $O\Delta(\xi_\alpha)\tilde{\psi}_\alpha=\varepsilon\Delta(\xi_\alpha)\tilde{\psi}_\alpha$. Inserting $\Delta(\xi_\alpha)^{-1}$ from the left, since the commutation with $O$ is guaranteed, we arrive at $O\tilde{\psi}_\alpha=\varepsilon\tilde{\psi}_\alpha$. Hence, the action of $O$ in exotic spinors is inherited from the usual spinor with the same eigenvalue. As it is clear, the argument is also valid backward. This straightforward construction shows, in particular, that usual and exotic spinors have the same chirality.   

Going further, consider the Dirac operator\footnote{In our convention, the mass term does not make part of $\textfrak{D}_0$. Moreover, $\textfrak{D}_0$ may also be applied to $\tilde{\Omega}(\mathcal{M})|_{U_\alpha}$ sections.} $\textfrak{D}_0:\Omega(\mathcal{M})|_{U_\alpha}\rightarrow \Omega(\mathcal{M})|_{U_\alpha}$. As remarked, Eq. (\ref{h2}) shows that $\Delta(\xi_\alpha)\tilde{\psi}_\alpha$ transforms as a local component of $\Omega(\mathcal{M})|_{U_\alpha}$. Therefore
\begin{equation}
\textfrak{D}_0\psi_\alpha=\textfrak{D}_0\Delta(\xi_\alpha)\tilde{\psi}_\alpha+\Delta(\xi_\alpha)\textfrak{D}_0\tilde{\psi}_\alpha,\label{pp1}
\end{equation} where $[\Delta(\xi_\alpha),\Gamma]=0$ was used in the matrix part of the Dirac operator. Eq. (\ref{pp1}) brings an important consequence. The derivative Dirac operator gives another local spinor when acting upon a local spinor. Therefore, $\textfrak{D}_0\tilde{\psi}_\alpha$ should also transform as Eq. (\ref{h2}) dictates, but the first term in the right-hand side of (\ref{pp1}) prevents it to occur, except for the trivial cases where $\Delta(\xi_\alpha)$ are constant, i.e., the cases of trivial topology. To restore the proper transformation, we shall proceed similarly to what is done with gauge potentials, bearing in mind that in the case at hand, we are dealing with a topological connection that cannot be absorbed by a gauge transformation (as we shall see). Let $\textfrak{B}_{(\alpha)}:\tilde{\Omega}(\mathcal{M})|_{U_\alpha}\rightarrow \tilde{\Omega}(\mathcal{M})|_{U_\alpha}$ be a general form taking values on $M(4,\mathbb{C})$ and define
 $\textfrak{D}:\tilde{\Omega}(\mathcal{M})|_{U_\alpha}\rightarrow \tilde{\Omega}(\mathcal{M})|_{U_\alpha}$ by 
\begin{equation}
\textfrak{D}:=\textfrak{D}_0+\textfrak{B}_{(\alpha)}.\label{d1}
\end{equation} Subtracting Eq. (\ref{pp1}) from $\Delta(\xi_\alpha)\textfrak{D}\tilde{\psi}_\alpha$ yields
\begin{equation}
\Delta(\xi_\alpha)\textfrak{D}\tilde{\psi}_\alpha-\textfrak{D}_0\psi_\alpha=[\Delta(\xi_\alpha)\textfrak{B}_{(\alpha)}-\textfrak{D}_0\Delta(\xi_\alpha)]\tilde{\psi}_\alpha. \label{ddd}
\end{equation} Imposing the left-hand side of Eq. (\ref{ddd}) to vanishes, leading to the correct transformation rule for the derivative of $\tilde{\psi}_\alpha$, we are able to find out the $\textfrak{B}_{(\alpha)}$ topological connection as simply $\textfrak{B}_{(\alpha)}=\Delta(\xi_\alpha)^{-1}\textfrak{D}_0\Delta(\xi_\alpha)$ and, as a consequence, to determine the $\textfrak{D}$ operator. Let us focus on the $\xi_\alpha$ functions and, via a bottom-up approach, arrive at the exotic spinor kinematic operator. 

First, observe that in light of our discussion so far, $\Delta(\xi_\alpha)$ can be realized by $\Delta(\xi_\alpha)=e^{i\theta_\alpha}\mathbb{1}_{M(4,\mathbb{C})}$, where $\theta_\alpha\equiv \theta_\alpha(x)$ for $x\in U_{\alpha}$ denotes a real function. Hence, $\Delta(\xi_\alpha)$ is unitary rendering invariance of the bilinear covariants. Moreover, $\xi_\alpha=e^{i\theta_\alpha}$ and the `M\"obiusity' encoded in Eq. (\ref{h2}) is traduced by 
\begin{equation}
\theta_\alpha-\theta_\beta=n\pi, \hspace{.2cm} n\in\mathbb{Z},\label{ar}
\end{equation} for $x\in U_\alpha\cap U_\beta$. It is straightforward to see that for $n$ odd, $\theta_\alpha(x)$ is discontinuous as $x$ cross a hypersurface $\Sigma \subset U_\alpha\cap U_\beta$, see Fig. (\ref{fig:tik}).    
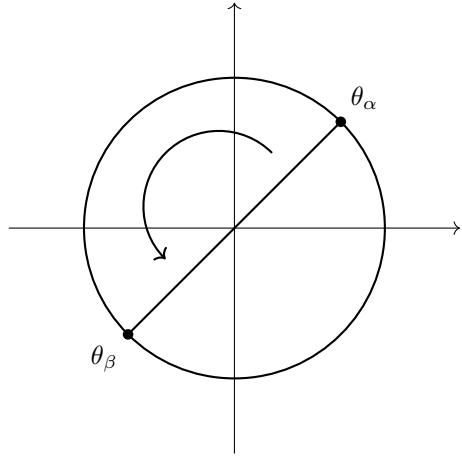
\begin{figure}[h]
	\centering
\begin{tikzpicture}
\draw[thick] (0,0) circle(2);

\draw[thick] (-1.414213562,-1.414213562) -- (1.414213562,1.414213562);

\draw[->] (-3,0) -- (3,0); 
\draw[->] (0,-3) -- (0,3); 

\draw[->, thick] (.5,1) arc[start angle=45, end angle=225, radius=1];

\node at (2,2,0.7) {$\theta_\alpha$};
\node at (-2,-2,-0.7) {$\theta_\beta$};

\fill[black] (1.414213562,1.414213562) circle (2pt);
\fill[black] (-1.414213562,-1.414213562) circle (2pt);
l
\end{tikzpicture}
\caption{Discontinuity of $\theta_\alpha(x)$ as $x$ cross the hypersurface $\Sigma\subset U_\alpha\cap U_\beta$.}
\label{fig:tik}
\end{figure}
Observe that $\xi_\alpha^2=\xi_\beta^2$ in $U_\alpha\cap U_\beta$ and the discontinuity is bypassed by the squares of $\xi_\alpha$ functions. We shall denote $\xi_\alpha^2:=\xi$ $\forall \alpha$ and, in the jargon, the `M\"obiusity' is said to be generated by the local square roots $\xi_\alpha$ of $\xi$. Of course, the impact of squaring $\xi_\alpha$ is traduced in Eq. (\ref{ar}) by $2\theta_\alpha\equiv 2\theta_\beta+2n\pi$. From now on, this exponent shall be expressed simply as $\theta(x)$ without irrelevant reference to the open set. So far, we have not specified the base manifold geometry. Now, to give up the open set dependence, we shall particularize the standard Dirac operator to the simplest case (flat spacetime case\footnote{As expected, there is no essential modification for the generalization to the curved base manifold formalism, except that the topological connection enters the spin connection.}) in which $\mathfrak{D}_0=i\gamma^{\mu}\partial_\mu$, where $\gamma^\mu$ stands for the Dirac matrices and $\mu=0,1,2,3$ in use to Einstein sum convention. Notice that this particularization should not be taken as a naive simplification. Instead, it is a straightforward way to ensure the curvature-free limit by construction. Within this case, the complete Dirac operator reads\footnote{Let emphasize again that $\mu$ index is summed up, while $\alpha$ stands for the open set $U_\alpha$ set.}
\begin{equation}
\textfrak{D}=i\gamma^\mu\partial_\mu+i\xi_\alpha^{-1}(x)\gamma^{\mu}\partial_\mu\xi_\alpha(x).\label{arr}
\end{equation} It can be readily verified that $\xi_\alpha^{-1}(x)\partial_\mu\xi_\alpha(x)=\xi^{-1}(x)\xi_\alpha(x)\partial_\mu[\xi(x)\xi^{-1}_\alpha(x)]$ yielding $\xi_\alpha^{-1}(x)\partial_\mu\xi_\alpha(x)=\frac{1}{2}\xi^{-1}(x)\partial_\mu\xi(x)$, and, thus, Eq. (\ref{arr}) reads 
\begin{eqnarray}
\textfrak{D}=i\gamma^\mu\partial_\mu-\frac{1}{2}\gamma^\mu\partial_\mu\theta(x)\label{cdo}
\end{eqnarray} The relative sign is not very relevant. It could be different had we chosen differently the relative sign in Eq. (\ref{d1}) or even absorbed it in the $\theta(x)$ definition.       

Let us further investigate the complete Dirac operator topological term, evincing several relevant aspects of it. It is convenient to evince an example of a topologically nontrivial base manifold. One could think of it as $\mathbb{R}^{1,3}\simeq\mathbb{R}\times(\mathbb{R}^2\times S^1)$. Consider $\{e_\mu\}$ the canonical base of $\mathbb{R}^{1,3}$ and $\{dx^\mu\}$ a base of $(\mathbb{R}^{1,3})^*$, such that $dx^\mu(e_\nu)=\delta_\nu^\mu$. 
\begin{enumerate}
\item {\it Winding number}: as well known, the linear Clifford map $\gamma:\mathbb{R}^{1,3}\rightarrow Cl_{1,3}$ furnish Dirac matrices when acting upon the canonical basis, i.e., $\gamma(e_\mu)=\gamma_\mu$. On the other hand, the Minkowski bilinear form is responsible for establishing a canonical isomorphism $\mathbb{R}^{1,3}\simeq (\mathbb{R}^{1,3})^*$. Therefore, $\gamma^\mu\partial_\mu \theta(x)=\eta^{\mu\nu}\gamma_\nu\partial_\mu\theta(x)$, or $\gamma(\eta^{\mu\nu}e_\nu\partial_\mu\theta(x))$, from which a composite mapping $\gamma\circ\eta:(\mathbb{R}^{1,3})^*\rightarrow Cl_{1,3}$ can be read. Thus, defining the $1-$form $\textfrak{b}:=d\theta=\partial_\mu\theta(x) dx^\mu$, the resulting $\textfrak{B}$ term can be recast as $\textfrak{B}=-\frac{1}{2}\gamma(\textfrak{b})$ and the Clifford map preimage defines the topological winding number
\begin{equation}
\frac{1}{2\pi}\oint_{S^1}\textfrak{b}=n \in \mathbb{Z}. 	
\end{equation} 
\item {\it Gauge transformation $\times$ topology}: the described procedure to find out the topological connection term may lead to the erroneous impression that it can be absorbed by a gauge transformation, given its resemblance with the standard compensating field. However, this elimination cannot be accomplished by any gauge transformation. In fact, coupling the exotic spinor to the electromagnetic potential, the complete Dirac operator gains, as well-known, an additional term $-e\slashed{A}$, where $e$ is the electrical charge absolute value and $\slashed{q}:=\gamma^\mu q_\mu$. On the other hand, the gauge transformation is $\slashed{A}\mapsto\slashed{A}+\frac{1}{e}\slashed{\partial}\Lambda(x)$, where $\Lambda:\mathbb{R}^{1,3}\rightarrow \mathbb{R}$ is an harmonic, smooth, and continuous function. In this vein, it is readily obtained that an attempt to gauge away the topological connection leads to
\begin{equation} 
\slashed{\partial}\Bigg(\Lambda(x)+\frac{e\theta(x)}{2}\Bigg)=0,     
\end{equation} implying $\Lambda(x)=cte-e(\theta(x)/2)$. However, remember that $\theta(x)/2$ is to be identified with $\theta_\alpha$ (for instance), entering the local square roots. Hence, $\Lambda(x)$  would have to be sensitive to the cocycle, or, putting it in dramatic words, it would have to be discontinuous (see Fig. (\ref{fig:tik})), a behavior not allowed for a gauge parameter. Ultra locally, i.e., in a neighborhood $\mathcal{V}\subset U_\alpha$ of $x$ not intercepted by $\Sigma$, the gauging away can be done, but apart from very restrictive conditions like this, the gauging away cannot be performed. 
\item {\it Nonequivalent solutions}: as remarked previously, exotic spinors are not equivalent to the usual ones. This fact can be trivially recovered from the Dirac operator's solutions point of view. The usual Dirac operator solutions belong to $\ker(\mathfrak{D_0})$, while the exotic counterpart belongs to $\ker(\mathfrak{D})$. Apart from the fact that these spaces are each subspace of a different spinor bundle, from (\ref{cdo}), they can only coincide if $\theta(x)$ is an everywhere constant function, i.e., there is no topological correction.  
\item {\it Lorentz symmetry breaking and physical consequences}: due to the most straightforward reasons, the complete Dirac operator (\ref{cdo}) does not lead to a covariant equation of motion. There is no aspect in the formalism dictating the $\theta(x)$ function as a scalar (in the Lorentz sense) field. Thus, $\partial_\mu\theta(x)$ does not transform as a vector component under Lorentz transformations. Let us explore this aspect via different physical perspectives. 

Starting from the equation of motion for the massive case, its `square' is given by $(\textfrak{D}+m)(\textfrak{D}-m)\tilde{\psi}=0$, leading to 
\begin{equation}
-\Box\tilde{\psi}-\frac{i}{2}\gamma^{\nu}\gamma^{\mu}\partial_\nu(\partial_\mu\tilde{\psi})-\frac{i}{2}\gamma^\nu\gamma^\mu\partial_\nu\theta(x)\partial_\mu\tilde{\psi}+\frac{1}{2}\partial^\mu\theta(x)\partial_\mu\theta(x)\tilde{\psi}-m^2\tilde{\psi}=0.\label{g1} 
\end{equation} By properly equating the second and third terms of Eq. (\ref{g1}), we arrive at 
\begin{equation}
(\Box+m^2)\tilde{\psi}+\textfrak{F}(x)\tilde{\psi}+i\partial^{\mu}\theta(x)\partial_\mu\tilde{\psi}=0,\label{g2}
\end{equation} where $\textfrak{F}(x)=\frac{1}{2}\Big\{i\Box\theta(x)-\frac{1}{2}\partial^\mu\theta(x)\partial_\mu\theta(x)\Big\}$. Recall that $\theta(x)$ is a real function. Therefore, there is no exact solution for $\theta(x)$ leading to $\textfrak{F}(x)=0$ other than $\theta(x)$ constant. This property profoundly impacts the relativistic dispersion relation since it directly introduces nonlinear dissipative terms in the frequency modes, a behavior usually associated with quasi-particles \cite{CE,CM}. It is reasonable, however, taking advantage of some freedom in the physical context to suppose that the effects of Lorentz violation terms are small enough so that $(\partial\theta)^2$ and $\partial^2\theta$ can be neglected in a good approximation. In fact, given the very stringent constraint coming from experiments (for a very comprehensive review, see Ref. \cite{vl}), disregarding $\textfrak{F}(x)$ is demanded. The Lorentz-violating term we are dealing with here comes from a different source of the current literature, but of course, it must have a magnitude compatible with experiments. Even disregarding the contributions of $\textfrak{F}(x)$, interesting effects appear. Calling $v^\mu=\partial^\mu\theta$ and taking the spinorial Fourier transform of $(\Box+m^2)\tilde{\psi}+iv^\mu\partial_\mu\tilde{\psi}\approx 0$, we have 
\begin{eqnarray} 
E^2+v^0E-f({\bf p})\approx 0,
\end{eqnarray} where $f({\bf p})=|{\bf p}|^2+m^2+{\bf v}\cdot{\bf p}$. Notice that since $|{\bf v}|\ll 1$, $f({\bf p})$ is always positive. Besides, for a reference frame at rest with respect to the particle (here understood in a semiclassical context), one has $E^2+v^0E\approx m^2$, indicating that setting $v^0=0$ is necessary to reproduce the usual relativistic rest energy. Therefore, the energy is given by 
\begin{equation}
E\approx \pm\sqrt{|{\bf p}|^2+m^2}\Big(1+\frac{{\bf v}\cdot{\bf p}}{|{\bf p}|^2+m^2}\Big)^{1/2} \label{ee0}
\end{equation} and retaining only the first order of the binomial expansion, we are left with 
\begin{equation}
E\approx E_0+\frac{{\bf v}\cdot{\bf p}}{2E_0},\label{eeq}
\end{equation} where $E_0=\pm\sqrt{|{\bf p}|^2+m^2}$. Eq. (\ref{eeq}) shows a small (but theoretically evident) difference in the dispersion relation of an exotic particle when contrasted to a usual one. As a complement, observe that the group velocity $v_g=\partial E/\partial |{\bf p}|$ obtained from (\ref{eeq}) is given by 
\begin{equation}
v_g=\frac{|{\bf p}|}{E_0}\Big(1-\frac{|{\bf v}||{\bf p}|\cos{\alpha}}{E_0^2}\Big)+\frac{|{\bf v}|\cos{\alpha}}{2E_0},
\end{equation} where $\alpha$ is the angle between ${\bf p}$ and ${\bf v}$. As it can be readily seen, in the used approximation scope, $v_g<1$, indicates the preservation of the causal structure.

Another consequence encoded into the topological extra term in the Dirac term may be brought to light by the standard Gordon decomposition \cite{Gor}. Hence, starting from $\mathfrak{D}\tilde{\psi}=0$ we have 
\begin{equation}
\gamma^\mu\partial_\mu\tilde{\psi}=\frac{-i}{2}\gamma^\mu\partial_\mu\theta(x)\tilde{\psi}-im\tilde{\psi} \label{ol1}
\end{equation} and, after taking the adjoint, a straightforward rearrangement leads to 
\begin{equation}
\partial_\mu\bar{\tilde{\psi}}\gamma^\mu=\frac{i}{2}\partial_\mu\theta(x)\bar{\tilde{\psi}}\gamma^\mu+im\bar{\tilde{\psi}}.\label{ol2}
\end{equation} Using Eqs. (\ref{ol1}) and (\ref{ol2}) the current conservation $\partial_\mu j^\mu=0$ is immediately obtained. Using Eqs. (\ref{ol1}) and (\ref{ol2}) along with the Clifford algebra relation $\{\gamma^\mu,\gamma^\nu\}=2\eta^{\mu\nu}\mathbb{1}$, where $\mathbb{1}$ stands for the $4\times 4$ identity matrix, it is possible to write 
\begin{eqnarray}
2m\bar{\tilde{\psi}}\gamma^\mu\tilde{\psi}=i[\bar{\tilde{\psi}}\gamma^\mu\gamma^\nu\partial_\nu\tilde{\psi}-\partial_\nu\bar{\tilde{\psi}}\gamma^\nu\gamma^\mu\tilde{\psi}]-\partial^\mu\theta(x)\bar{\tilde{\psi}}\tilde{\psi}\mathbb{1}.\label{oll}
\end{eqnarray} Calling, as usual, $\sigma^{\mu\nu}=\frac{i}{2}[\gamma^\mu,\gamma^\nu]$, Eq. (\ref{oll}) can be re-expressed in such a way that the conserved current reads
\begin{eqnarray}\label{off}
j^\mu=\frac{i}{2m}(\bar{\tilde{\psi}}\partial^\mu\tilde{\psi}-\partial^\mu\bar{\tilde{\psi}}\tilde{\psi})\mathbb{1}+\frac{1}{2m}\partial_\nu(\bar{\tilde{\psi}}\sigma^{\mu\nu}\tilde{\psi})-\frac{1}{2m}v^\mu\bar{\tilde{\psi}}\tilde{\psi}\mathbb{1}. 
\end{eqnarray} A coupling with the electromagnetic potential $-ej^\mu A_\mu$, where $e$ denotes the electric charge absolute value, is insightful. The interaction Hamiltonian contains a rather standard part encoded in the second term of Eq. (\ref{off}) leading to the magnetic moment factor after standard textbook quantization and non-relativistic limit \cite{PS,Sak}, but the last term of (\ref{off}) encodes an effective term violating Lorentz symmetry in the vertex function
\begin{equation}
\Gamma^\mu=-\frac{1}{e}\frac{\delta^3 S_{eff}}{\delta\bar{\tilde{\psi}}\delta\tilde{\psi}\delta A_\mu}. 
\end{equation} The classic limit of $\Gamma^\mu$ shall, then, go as $\Gamma^\mu|_{Classic}=\gamma^\mu+v^\mu\mathbb{1}$ leading to a slight deviation in the foton-exotic spinorial particle vertex\footnote{It is curious that the leading-order contribution to $\Gamma^\mu$ comes with a typical central charge-like term, and all the effects had been mathematically originated in an exponential ``phase'' relation between spinors (\ref{h2}), in a {\it pari passu} association to projective representations whose roots are in the algebraic aspects.}.   
   
\item {\it Exoticity as a spinorial characteristic}: as it is clear from the previous sections' constructions, the possibility of an exotic counterpart is exclusive to spinor fields. The comprehensive exposition adopted, however, may obscure this aspect. Recall that a central aspect of the existence of exotic spinors is the possibility of a different (nonequivalent) spinor structure. The $\delta_{\alpha\beta}$ cocycle plays, then, the role of connecting different usual and exotic spinors, as seen by Eqs. (\ref{aa}), (\ref{shall}), and (\ref{h2}). The net effect of $\Delta(\xi)$, inherited from $\xi_\alpha/\xi_\beta=\pm 1=\Delta(\delta_{\alpha\beta})$, is to relate local sections spinors in different spinor bundles locally. The negative sign is a peculiar characteristic of spinorial transformations. A tensorial field never reaches this peculiarity. In this regard, it is important to emphasize that usual spinorial objects entering in expansion coefficients of quantum fields, along with a different spinor dual, may lead to a quantum bosonic statistic \cite{DCY}, with the spin-statistic theorem premise being circumvented by a pseudo-hermitian formulation \cite{most,most2,lec}. That is to say, exoticity seems to be a characteristic of spinors, not necessarily fermions.    

\item {\it An attempt to reinterpret the topological Dirac operator correction}: from Eq. (\ref{cdo}), in natural units, we see that the additional topological term in the Dirac operator scales as $length^{-1}$. This behavior suggests that the topological term effect shall be appreciated at high energy scales (or equivalently short length scales) and, in this case, at higher energy enough, spacetime (or three-space) topology shall be such that exotic spinors arise. The scenario resembles some of the spacetime foam frameworks (see Fig. \ref{fifig} and Ref. \cite{stf} for a review).   

\begin{figure}[h]
	\begin{center}
		\includegraphics[scale=.3]{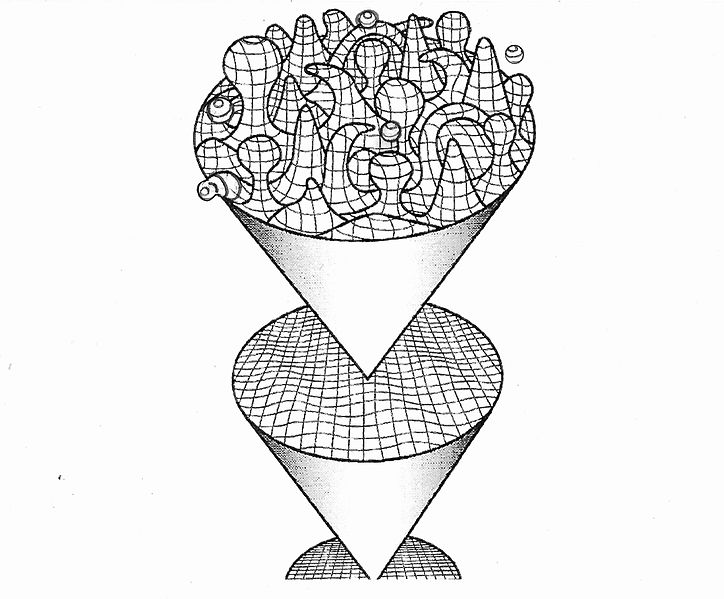}
		\caption{Spacetime fluctuation `{\it Temaki}' representation, including the possibility of topology change at high energies. Figure obtained from {\it Wikipedia} (https://en.wikipedia.org/wiki/Quantum\textunderscore foam).}
	\end{center}\label{fifig}
\end{figure}

We shall face this item as a deliberate attempt to speculate about an interpretation in this sense. We shall not exhaust the question here, nor even advocate well-posed partial results, but instead, explore some points in such a direction. Firstly, certain rules exist for a given manifold to properly accommodate different topologies. The general idea is to allow for the manifold to smoothly interpolate between low dimensional manifolds, say $\mathcal{M}_1$ and $\mathcal{M}_2$. The intuitive way for implementing it (also given the discussion around Eq. (\ref{ee0})) is to allow a topology change in the spatial section of $\mathcal{M}$ only. The basic idea is to allow for $\partial\Sigma=\mathcal{M}_1 \uplus \mathcal{M}_2$, where $\mathcal{M}\supset\Sigma$ and $\uplus$ stands for the disjoint union. There are several possibilities for implementing the cobordism \cite{TA}, many of them quite involved from the mathematical point of view. More often than never, in physics, it is discussed the possibility of a Lorentzian cobordism between closed (compact without boundary) spaces\footnote{Such an approach has topological advantages, as well-defined connected sums \cite{rob}.}. Further pursuing this possibility, $\Sigma$ is a four-dimensional submanifold endowed with a Lorentzian metric whose boundary is given by the disjoint union of the two three-dimensional closed spacelike hypersurfaces $\mathcal{M}_1$ and $\mathcal{M}_2$. There is a fundamental theorem \cite{novomilnor} shows that a cobordism exists if, and only if, the Stiefel-Whitney classes of $\mathcal{M}_1$ and $\mathcal{M}_2$ are equal. Contrasting this theorem to our case, since it is desirable to keep orientability and we must keep spinor structures allowed, we have $w_1=0=w_2$ and consequently $w_3=0$ for both $\mathcal{M}_1$ and $\mathcal{M}_2$. However, this approach has a caveat: it was shown in Ref. \cite{outger} that a Lorentzian cobordism leads to closed timelike curves. The central idea is that a Lorentzian cobordism interpolating between $\mathcal{M}_1$ and $\mathcal{M}_2$ allows causal connections between these hypersurfaces. Nevertheless, topology, by itself, bends light cones. Therefore, events of one hypersurface may influence events in another hypersurface and vice-versa\footnote{Think of the light cones bending along the cobordism.}. This problem can be avoided by introducing a degenerate metric at specific points of the interpolating manifold \cite{kin} to avoid the existence of causal trajectories.  

Within this context, since the topology at higher energy scales is supposed to be non-simply connected, the quantization of fields is also tricky. The adjacent spaces $\mathcal{M}_1$ and $\mathcal{M}_2$ generally define different vacuum related by a Bogoliubov transformation \cite{And,trore}. A given field experiencing these scales necessarily undergoes a compactification (or decompactification) in at least one dimension, and infinitely many particles may be created from this process. To avoid this undesired behavior, one must apply some renormalization technique (depending on the specific case and physical interest) \cite{pd,BD}. 

There is yet another layer to be considered in the context described here. The spin structures built upon $\mathcal{M}_i$ ($i=1,2$) must be extended to $\Sigma$. The vanishing of $w_2$ entails this property, but it is also necessary that a spinor structure cobordism also happen. That is to say, it must exist a spinor structure in $\Sigma$, say $P_\Sigma$, which is restricted to the spinor structure of $\mathcal{M}_i$ at $\mathcal{M}_i\cap\partial\Sigma$ \cite{cs1}. Besides, within the spinor exotic program, more than one compatible spinor structure cobordism may be in order, and further analysis about the possible physical implications is demanded \cite{cs2}. All these aspects are mentioned here to give at least a superficial account of the challenges in this line of reasoning. 

\end{enumerate}

\section{Overview} 

For organization purposes, the developments in this field may be separated into two branches: one dealing with advances reached in the seventies, eighties, and nineties, and more or less recent achievements. Let us report on both periods.  

\subsection{Early advances}

Perhaps the relation between non-simply connectivity and its impact on Dirac operators may be traced back to Ref. \cite{At1} in the formal context of index theorems \cite{At2}. In parallel, multiply-connected manifolds were used in studying Yang-Mills theory applications \cite{sch,sas}, and quantum fields at finite-temperature were also investigated in such backgrounds \cite{uno,duno}. In Ref. \cite{petry}, an exciting application of nonequivalent spin structures, applying exotic spinors to the Bardeen-Cooper pair formation in a space-section $S^1\times \mathbb{R}^2$ model for superconductivity. In the condensed matter framework, the electron-electron pair interaction in the fractional quantum Hall effect context was suggested to be interpreted using exotic fermions \cite{hesshall}.

In an incomplete account to arrive at exotic spinors being systematically investigated in the quantum realm, we emphasize that while different topologies became relevant in the quantum/gravitational context \cite{ag}, spinor structures also were investigated in this framework \cite{bb}. In Ref.\cite{HP}, it is pointed out that some spaces studied in a quantum gravity theory do not even allow spinor fields. However, generalized spin structures may be relevant in connecting topology to particle spectrum. This peculiar formulation was also further mathematically explored in Ref. \cite{dal}. An investigation of quantum fields whose dynamics take place in nontrivial spacetime topology backgrounds was performed in Ref. \cite{royal}, where it was suggested that topology could determine the size of a supermultiplet composed by a type of fields possible to be there defined, the twisted\footnote{Essentially, the twisted field is a usual field endowed with anti-periodic boundary conditions.} fields. A similar, though different, idea is explored in Ref. \cite{supsup} where supermultiplets (including fermions) are studied in a supergravity framework.  

Narrowing the analysis to exotic fermions, in Ref. \cite{sf}, the Dirac lagrangian is considered in detail. It is also in this paper that the term `M\"obiusity,' coined by Salam, is firstly used to refer to the number of elements of $\check{H}^1(\mathcal{M},\mathbb{Z}_2)$ which, by its turn, is shown to be the number of different lagrangians. The correspondence between different exotic spinor structures, Dirac operators, and, therefore, Dirac lagrangians may be straightforwardly traced in our preceding sections. As mentioned in the text, Ref. \cite{AI} formulates (and solves) the problem of Lorentz invariance breaking by the vacuum generating functional for spinor fields in multiple connected spaces. Twisted and untwisted spinor fields were extensively studied in Ref. \cite{ford} in the context of an $S^1\times \mathbb{R}^3$ spacetime. The selection of an ulterior spinor possibility leads to a well-defined vacuum polarization upon vacuum propagation, while untwisted spinor fields suffer from non-causal effects. Notice that in the cases here mentioned of twisted fields \cite{supsup,ford}, the background also admits exotic spinors. Contrasting the possibility of twisting with exotic spinor fields would be very important. In particular, for Ref. \cite{ford}, the union of exoticity with twisted fields shall be challenging since the twist possibility comes from a very restrictive (vanishing connection) constraint upon the Dirac operator. To our knowledge, analysis verifying whether an exotic twisted spinor exists (and its physical consequences) is an open issue. It is also relevant to stress that exotic Majorana spinors, real spinors from a nonequivalent spin structure, can also be consistently defined \cite{maje}. Besides, different spin structures' impact on string theory was also taken into account \cite{SW}.    

\subsection{More recent achievements}

There was a time gap between the early and more recent advances in the exotic spinor fields. This hiatus may be partially imputed to the hermetic character of mathematical necessary preliminaries, usually not directly reported in detail (something that this review intends to mitigate), and also to the Lorentz symmetry violation present in the complete Dirac operator topological term in multiply connected spacetimes. From the $2000$ year on, a renewed interest in the Lorentz symmetry violation arose, and new work has been done.     

The analysis of an exotic dark matter candidate (from the classical point of view) was done in detail in Ref. \cite{nn} with due care to the mathematical preliminaries. Continuing the program of Ref. \cite{nn}, taking advance of the spinor field neutrality, the underlining nontrivial topology was bounded to the dark spinor field behavior Ref. \cite{prob}. Besides, the dispersion relation to such a dark field was investigated in Ref. \cite{dis}. Although the geodesic completeness program can be applied to a certain class of black holes, in a given context, these singularities may be faced as generators of nontrivial topology, so to speak. This possibility was scrutinized in Refs. \cite{bh1} and \cite{bh2}, where an intuitive approach based on Cartan's spinorial point of view is pursued. Moreover, the impact of the Hawking radiation topological terms is studied. 

The investigation of a specific subclass of Inomata-McKinley spinors and their exotic counterpart's (im)possibility is scrutinized in Ref. \cite{dino}. Exotic spinors were also studied in a background (and corresponding algebra) endowed with a minimal length \cite{min}. In particular, the topological term in the Dirac operator prevents it from being injective (with minimal length correction taken into account), and solutions other than the trivial one are expected. Ref. \cite{hear} shows that the topological term entering the Dirac operator leads to the heat kernel coefficients corrections. Hence, the spectral properties of the complete Dirac operator and geometric invariants of the base manifold. 

A type of nontrivial topology geometrization effects, including finding a symmetric, bilinear, and (somewhere) non-degenerate metric, can be found in Ref. \cite{car}. The geometrization is obtained by combining exoticity and the classic Cartan's perspective of spinors as the square root of spacetime points \cite{livrocar,pen}. It was obtained so that the $k$ differential form coefficients are also deformed. This last aspect leads to deformations in electromagnetism \cite{mono}, while a deformed metric allows for new spinor classes \cite{ptep}. Finally, it must be emphasized that some of the previous accomplishments were reached in the scope of a spinor bundle deformation whose detailed account may be found in Ref. \cite{def}, where, after a complete mathematical characterization, it is shown that in a specific context, exotic spinors may be used to compute the neutron magnetic moment without worrier about its internal structure.    

Ending this section by pointing to some possible lines of further exploration is essential. In this regard, it is tempting to mention the results of Ref. \cite{muni} again. Even though the study relies on non-vanishing three-curvature, a polemical issue, it would be very relevant to see whether contributions to the spinorial energy-momentum tensor coming from exotic fermions play a role in the accelerated universe expansion. As previously mentioned, it would be important to contrast twisted with exotic fields, and combining these possibilities in cosmological scenarios would seem fruitful. Besides that, we have discussed in Section III.A. the action of diffeomorphisms on manifolds (allowing for several spin bundles) as a permutation between spin structures. We believe a deeper understanding of this behavior can be achieved using braid groups \cite{aldro}.    

From another perspective, the project started in Refs. \cite{petry,hesshall} continuation is also demanding. In particular, the impact of exotic spinors, if any, in many-body systems in manifolds endowed with nontrivial topology \cite{hub} and in odd topological phases in a superconductivity and noncollinear magnetism-bound setup \cite{z2}. When adapted to (and contrasted with) condensed matter systems, the formalism underlying exotic spinors allows for a better interpretation, a characteristic highly desired in any physical branch of research. As a last remark on possible investigation directions, we call attention to the fact that the studied formal aspects are almost directly transferable to spin $3/2$ fields. Nevertheless, there is a caveat in studying the physical effects of exotic Rarita-Schwinger fields' counterparts \cite{rar}.    

\section*{Acknowledgments}

The author thanks CNPq (grant No. 307641/2022-8) for financial support. It is a pleasure to thank Prof. Roldao da Rocha for carefully reading the manuscript and stimulating discussions.

\end{document}